\documentclass[12pt]{article}
\usepackage{amsmath,amsthm,latexsym,amssymb,amsfonts,epsfig}
\usepackage{color}
\usepackage{ulem}
\usepackage[T1]{fontenc}
\usepackage{amssymb,graphicx}

\usepackage[all,arc,dvips,arrow,curve,cmactex]{xy}
\usepackage{amsfonts,amsmath,amssymb,graphics,psfrag}
\def\SSEarrow{\ensuremath{\rotatebox[origin=c]{-90}{$\twoheadrightarrow$}}}
\def\Ssearrow{\ensuremath{\rotatebox[origin=c]{90}{$\twoheadrightarrow$}}}

\newcommand\Sets{{\bf Sets}}
\newcommand\op{{\rm op}}
\newcommand\ps[1]{\underline{#1}}

\newcommand\ts{\tilde{\Sigma}}
\newcommand\ic{\mathbf{I}\Cl}

\makeatletter
\@addtoreset{equation}{section}
\makeatother

\newtheorem{Theorem}{Theorem}[section]
\newtheorem{Definition}{Definition}[section]

\newtheorem{Corollary}{Corollary}[section]
\newtheorem{Proposition}{Proposition}[section]
\newtheorem{Conjecture}{Conjecture}[section]

\newcommand\Sig{\underline{\Sigma}}

\def\be{\begin{equation}}
\def\ee{\end{equation}}
\def\ba{\begin{eqnarray}}
\def\ea{\end{eqnarray}}
\def\a{{\cal A}}

\def\mv{\mathcal{V}}
\def\mh{\mathcal{H}}
\def\mb{\mathcal{B}}
\def\us{\underline{\Sigma}}

\def\mn{\mathcal{N}}

\def\Nl{{\mathchoice
{\setbox0=\hbox{$\displaystyle\rm N$}\hbox{\hbox to0pt
{\kern0.4\wd0\vrule height0.9\ht0\hss}\box0}}
{\setbox0=\hbox{$\textstyle\rm N$}\hbox{\hbox to0pt
{\kern0.4\wd0\vrule height0.9\ht0\hss}\box0}}
{\setbox0=\hbox{$\scriptstyle\rm N$}\hbox{\hbox to0pt
{\kern0.4\wd0\vrule height0.9\ht0\hss}\box0}}
{\setbox0=\hbox{$\scriptscriptstyle\rm N$}\hbox{\hbox to0pt
{\kern0.4\wd0\vrule height0.9\ht0\hss}\box0}}}}
\def\Zl{{\mathchoice
{\setbox0=\hbox{$\displaystyle\rm Z$}\hbox{\hbox to0pt
{\kern0.4\wd0\vrule height0.9\ht0\hss}\box0}}
{\setbox0=\hbox{$\textstyle\rm Z$}\hbox{\hbox to0pt
{\kern0.4\wd0\vrule height0.9\ht0\hss}\box0}}
{\setbox0=\hbox{$\scriptstyle\rm Z$}\hbox{\hbox to0pt
{\kern0.4\wd0\vrule height0.9\ht0\hss}\box0}}
{\setbox0=\hbox{$\scriptscriptstyle\rm Z$}\hbox{\hbox to0pt
{\kern0.4\wd0\vrule height0.9\ht0\hss}\box0}}}}
\def\Ql{{\mathchoice
{\setbox0=\hbox{$\displaystyle\rm Q$}\hbox{\hbox to0pt
{\kern0.4\wd0\vrule height0.9\ht0\hss}\box0}}
{\setbox0=\hbox{$\textstyle\rm Q$}\hbox{\hbox to0pt
{\kern0.4\wd0\vrule height0.9\ht0\hss}\box0}}
{\setbox0=\hbox{$\scriptstyle\rm Q$}\hbox{\hbox to0pt
{\kern0.4\wd0\vrule height0.9\ht0\hss}\box0}}
{\setbox0=\hbox{$\scriptscriptstyle\rm Q$}\hbox{\hbox to0pt
{\kern0.4\wd0\vrule height0.9\ht0\hss}\box0}}}}
\def\Rl{{\mathchoice
{\setbox0=\hbox{$\displaystyle\rm R$}\hbox{\hbox to0pt
{\kern0.4\wd0\vrule height0.9\ht0\hss}\box0}}
{\setbox0=\hbox{$\textstyle\rm R$}\hbox{\hbox to0pt
{\kern0.4\wd0\vrule height0.9\ht0\hss}\box0}}
{\setbox0=\hbox{$\scriptstyle\rm R$}\hbox{\hbox to0pt
{\kern0.4\wd0\vrule height0.9\ht0\hss}\box0}}
{\setbox0=\hbox{$\scriptscriptstyle\rm R$}\hbox{\hbox to0pt
{\kern0.4\wd0\vrule height0.9\ht0\hss}\box0}}}}
\def\Cl{{\mathchoice
{\setbox0=\hbox{$\displaystyle\rm C$}\hbox{\hbox to0pt
{\kern0.4\wd0\vrule height0.9\ht0\hss}\box0}}
{\setbox0=\hbox{$\textstyle\rm C$}\hbox{\hbox to0pt
{\kern0.4\wd0\vrule height0.9\ht0\hss}\box0}}
{\setbox0=\hbox{$\scriptstyle\rm C$}\hbox{\hbox to0pt
{\kern0.4\wd0\vrule height0.9\ht0\hss}\box0}}
{\setbox0=\hbox{$\scriptscriptstyle\rm C$}\hbox{\hbox to0pt
{\kern0.4\wd0\vrule height0.9\ht0\hss}\box0}}}}
\def\Hl{{\mathchoice
{\setbox0=\hbox{$\displaystyle\rm H$}\hbox{\hbox to0pt
{\kern0.4\wd0\vrule height0.9\ht0\hss}\box0}}
{\setbox0=\hbox{$\textstyle\rm H$}\hbox{\hbox to0pt
{\kern0.4\wd0\vrule height0.9\ht0\hss}\box0}}
{\setbox0=\hbox{$\scriptstyle\rm H$}\hbox{\hbox to0pt
{\kern0.4\wd0\vrule height0.9\ht0\hss}\box0}}
{\setbox0=\hbox{$\scriptscriptstyle\rm H$}\hbox{\hbox to0pt
{\kern0.4\wd0\vrule height0.9\ht0\hss}\box0}}}}
\def\Ol{{\mathchoice
{\setbox0=\hbox{$\displaystyle\rm O$}\hbox{\hbox to0pt
{\kern0.4\wd0\vrule height0.9\ht0\hss}\box0}}
{\setbox0=\hbox{$\textstyle\rm O$}\hbox{\hbox to0pt
{\kern0.4\wd0\vrule height0.9\ht0\hss}\box0}}
{\setbox0=\hbox{$\scriptstyle\rm O$}\hbox{\hbox to0pt
{\kern0.4\wd0\vrule height0.9\ht0\hss}\box0}}
{\setbox0=\hbox{$\scriptscriptstyle\rm O$}\hbox{\hbox to0pt
{\kern0.4\wd0\vrule height0.9\ht0\hss}\box0}}}}

\title{{\sf Complex Numbers, One-Parameter of Unitary Transformations and Stone's Theorem in Topos Quantum Theory}\\ 
}

\author{{\sf W. Brenna$^1$\thanks{{\sf wbrenna@uwaterloo.ca}},  C. Flori$^2$}\thanks{{\sf cflori@perimeterinstitute.ca}}\\
\\
{\sf $^2$ Perimeter Institute for Theoretical Physics,}\\
{\sf 31 Caroline Street N, Waterloo, ON N2L 2Y5, Canada}\\
{\sf $^1$ Department of Physics \& Astronomy, University of Waterloo}\\
{\sf 200 University Avenue West, Waterloo, Ontario, Canada, N2L 3G1 }}
\date{{\small\sf Preprint }}

\begin{document}
\maketitle
\begin{abstract}
Topos theory has been suggested first by Isham and Butterfield, and then by Isham and
D¬oring, as an alternative mathematical structure within which to formulate physical theories.
In particular, it has been used to reformulate standard quantum mechanics in such a way that
a novel type of logic is used to represent propositions. In recent years the topic has been considerably progressing with the introduction of probabilities, group and group transformations. In the present paper we will introduce a candidate for the complex quantity value object and analyse its relation to the real quantity value object. By defining the Grothendieck k-extension of these two objects, so as to turn them into abelian groups, it is possible to define internal one parameter groups in a topos. We then use this new definition to construct the topos analogue of the Stone's theorem.
\end{abstract}

\newpage\section{Introduction}
Recently, Isham and D\"oring   have developed a novel formulation of quantum theory based on the mathematical structure of topos theory,
first suggested by Isham and Butterfield, 
\cite{Isham1999a}, \cite{Isham1999},  \cite{Hamilton1999},  \cite{Butterfield2002},  \cite{Isham2000},  \cite{Doering2008}. 

The aim of this new formulation is to overcome the Copenhagen (instrumentalist) interpretation of quantum theory and replace it with an observer-independent, non-instrumentalist interpretation.

The strategy adopted to attain such a new formulation is to
re-express quantum theory as a type of `classical theory' in a
particular topos.
In this setting, the notion of classicality is defined in terms of the notion of {\it context} or {\it classical snapshots}.
In particular, in this framework, quantum theory is seen as a collection of local `classical snapshots',
where the quantum information is determined by the relation between these local classical snapshots.

Mathematically, each classical snapshot is represented by an
abelian von-Neumann sub-algebra $V$ of the algebra
$\mathcal{B}(\mh)$ of bounded operators on a Hilbert space. The
collection of all these contexts forms a category $\mv(\mh)$,
which is actually a poset by inclusion.
As one goes to smaller sub-algebras $V^{'}\subseteq V$ one obtains a coarse-grained classical perspective on the theory.

The fact that the collection of all such classical snapshots forms a category, in particular a poset, means that the quantum information can be retrieved by the relations between such snapshots, i.e. by the categorical structure.

A topos that allows for such a classical local description is the
topos of presheaves over the category $\mv(\mh)$. This is denoted
as $\Sets^{\mv(\mh)^{\op}}$. By utilising the topos $\Sets^{\mv(\mh)^{\op}}$ to
reformulate quantum theory, it was possible to define pure quantum
states, quantum propositions and truth values of the latter
without any reference to external observer, measurement or any
other notion implied by the instrumentalist interpretation. In
particular, for pure quantum states, probabilities are replaced by
truth values, which derive from the internal structure of the
topos itself.
These truth values are lower sets in the poset $\mv(\mh)$, thus they are interpreted as the collection of all classical snapshots for which the proposition is true.
Of course, being true in one context implies that it will be true in any coarse graining of the latter.

However, this formalism lacked the ability to consider mixed
states in a similar manner as pure states, in particular it lacked
the ability to interpret truth values for mixed states as
probabilities. This problem was solved in \cite{Doering2011} by
enlarging the topos $\Sets^{\mv(\mh)^{\op}}$ and considering, instead, the
topos of sheaves over the category $\mv(\mh)\times (0,1)_L$, i.e.
$Sh( \mv(\mh)\times(0,1)_L)$. Here $(0,1)_L$ is the category
whose open sets are the intervals $(0,r)$ for $0\leq r\leq 1$.
Within such a topos it is also possible to define a logical
reformulation of probabilities for mixed states. In this way
probabilities are derived internally from the logical structure of
the topos itself and not as an external concept related to
measurement and experiment. Probabilities thus gain a more
objective status which induces an interpretation in terms of
propensity rather than relative frequencies.

Moreover, it was also shown in \cite{Doering2011} that all that was done for the topos $\Sets^{\mv(\mh)^{\op}}$ can be translated to the topos $Sh( \mv(\mh)\times(0,1)_L)$.
Although much of the quantum formalism has been re-expressed in
the topos framework there are still many open questions and
unsolved issues. Of particular importance is the role of unitary
operators and the associated concept of group transformations. In
\cite{Doering2008a}, \cite{Doering2008b}, \cite{Doering2008c},
\cite{Doering2008d}, \cite{Doering2008} the role of unitary operators in
the topos $\Sets^{\mv(\mh)^{\op}}$ was discussed and it was shown that
generalised truth values of propositions transform `covariantly'.
However, the situation is not ideal since `twisted' presheaves had
to be introduced.
This problem was solved in \cite{Flori2011} where the authors define the notion of what a group and associated group transformation is in the topos representation of quantum theory, in such a way that the problem of twisted presheaves is avoided.
In order to do this they slightly change the topos they work with.
The reasons for this shift are: i) they require the group action to be continuous ii) In order to avoid twisted presheaves the base category has to be fixed, i.e. the group can not be allowed to act on it.

Although the problem of group transformations was solved, the precise definition of unitary operators and unitary transformations still remains open.
In particular, one may ask how unitary operators represented in the topos formulation of quantum theory.
We know that self-adjoint operators are represented as arrows from the state space $\us$ to the quantity value object $\ps{\Rl^{\leftrightarrow}}$.
The natural question to ask is whether such a representation can be extended to all normal operators.
To this end one needs to, first of all, define the topos analogue of the complex numbers.
Of course there is the trivial object $\ps{\Cl}$ but this, as we will see, can not be identified with the complex number object since 
a) it does not reduce to $\ps{\Rl^{\leftrightarrow}}$,
and b) since the presheaf maps in $\ps{\Cl}$ are the identity maps, these maps will not respect the ordering induced by the yet to be defined daseinisation of normal operators.
Thus, some other object has to be chosen as the complex valued quantity object.
In this paper we will define such an object.
In order to construct this object we will first of all define an ordering on the complex numbers $\Cl$, which is related to the ordering of the spectra of the normal operators induced by the ordering of the self-adjoint operators comprising them.
We then arrive at a definition of the complex number object $\ps{\Cl^{\leftrightarrow}}$ and, consequently, at a definition of normal operators as maps from the state space to the newly defined object.

Given that we now have both the complex and real quantity value objects, we attempted to define one parameter goup in terms on these objects, however, this was not possible since neither $\ps{\Cl^{\leftrightarrow}}$ nor $\ps{\Rl^{\leftrightarrow}}$ are groups, but are only monoid.
To solve this problem we applied the Grothendieck k-extension \cite{Doering2008} so as to obtain the abelian group objects $k(\ps{\Cl^{\leftrightarrow}})$ and $k(\ps{\Rl^{\leftrightarrow}})$.
To simplify the notation we switched to the objects $k(\ps{\Cl}^{\geq})\subseteq k(\ps{\Cl^{\leftrightarrow}}) $ and $k(\ps{\Rl}^{\geq})\subseteq k(\ps{\Rl^{\leftrightarrow}})$; this poses no loss in generality.
We then were able to define the topos description of a one parameter group taking values in $k(\ps{\Cl}^{\geq})$ and $k(\ps{\Rl^{\leftrightarrow}})$.
We then apply these topos analogue of one parameter group of transformation to define and proof the topos analogue of Stone's theorem.

\section{Possible Ordering of Complex Numbers}
\subsection{Spectral Theorem for Normal Operators}

In the current formalism of topos quantum physics \cite{Doering2008}, the spectral theorem is used when representing ``physical quantities''.
In fact, in this formulation self-adjoint operators are well-defined, as maps between the spectral presheaf and the quantity value object (see Section \ref{sec:topos} in the Appendix for the relevant definitions):
\be
\breve{\delta}(\hat{A}):\us\rightarrow\underline{\Rl}^{\leftrightarrow}
\ee
such that for each context $V\in\mv(\mh)$ we have
\ba
\breve{\delta}(\hat{A})_V:\us_V&\rightarrow&\underline{\Rl}^{\leftrightarrow}_V\\
\lambda&\mapsto&\big(\breve{\delta}^i(\hat{A})(\lambda), \breve{\delta}^o(\hat{A})(\lambda)\big)
\ea
where 
\ba
\breve{\delta}^i(\hat{A})(\lambda): \downarrow V&\rightarrow& \Rl\\
V^{'}&\mapsto&\lambda(\delta^i(\hat{A})_{V^{'}})
\ea
is an order preserving function and $\lambda(\delta^i(\hat{A})_{V^{'}})$ represents the value of the inner daseinised operator $\delta^i(\hat{A})_{V^{'}}$. On the other hand
\ba
\breve{\delta}^o(\hat{A})(\lambda): \downarrow V&\rightarrow& \Rl\\
V^{'}&\mapsto&\lambda(\delta^o(\hat{A})_{V^{'}})
\ea
is an order reversing function and $\lambda(\delta^o(\hat{A})_{V^{'}})$ represents the value of the outer daseinised operator $\delta^o(\hat{A})_{V^{'}}$ see \cite{Doering2008} for details. \\

We would like to pursue the same approach but for normal operators, so we begin by stating the spectral theorem for bounded normal operators \cite{Whitley1968,Helmberg2008}.
Since we know the spectral decomposition of a self-adjoint operator, we begin by breaking up the normal operator $C$ into two self-adjoint parts
\begin{align*}
B &= \frac{1}{2} \left( \hat{C} + \hat{C}^* \right) \\
A &= \frac{1}{2i} \left( \hat{C} - \hat{C}^* \right)
\end{align*}

This decomposition has a number of unfortunate downsides, of which the most important to us is the fact that daseinisation is not additive:
\begin{equation}
\label{nadditive}
\delta(\hat{A} + \hat{B} )_V \le \delta(\hat{A})_V + \delta(\hat{B})_V
\end{equation}

Therefore, directly generalizing, the definition of self-adjoint operator will not hold.
It can be seen, however, that the spectral decomposition can be better defined \cite{Helmberg2008}.
We know that normal operators have the representation
\be
\hat{C}=\int_{\Cl}\lambda d\hat{E}^{\hat{C}}_{\lambda}
\ee
However we also know that $\hat{C}=\hat{A}+i\hat{B}$ with $\hat{A}=\int_{\Rl}\gamma\hat{E}^{\hat{A}}_{\gamma}$ and $\hat{B}=\int_{\Rl}\sigma\hat{E}^{\hat{B}}_{\sigma}$, therefore
\be
\hat{C}=\int_{\Rl}\gamma d\hat{E}^{\hat{A}}_{\gamma}+\int_{\Rl}i\sigma d\hat{E}^{\hat{B}}_{\sigma}=\int_{\Rl}\left(\gamma+i\sigma \right) d\hat{E}^{\hat{A}}_{\gamma} d\hat{E}^{\hat{B}}_{\sigma}
\ee
So what exactly is the relation between those two expressions, and furthermore, what is the relation between $\lambda$ and $\gamma+i\sigma$?\\The answer can be found in the following theorem:

\begin{Theorem}
Given a bounded operator $\hat{A}=\hat{C}+i\hat{B}$ on $\mh$, there exists a family of projection operators 
$\{\hat{P}(\varepsilon, \eta):=\hat{P}_1(\varepsilon)\hat{P}_2(\eta)|(\varepsilon, \eta)\in\Rl^2\}$
which commute with $\hat{A}$, where $\{\hat{P}_1(\varepsilon)|\varepsilon\in \Rl\}$ is the spectral family of $\hat{C}$ and 
$\{\hat{P}_2(\eta)|\eta\in\Rl\}$ is the spectral family of $\hat{B}$.
We then say that $\{P(\varepsilon, \eta):=\hat{P}_1(\varepsilon)\hat{P}_2(\eta)|(\varepsilon, \eta)\in\Rl^2\}$ is the spectral family of $\hat{A}$.
Such a family has the following properties:
\begin{enumerate}
\item [a)]  $\hat{P}(\varepsilon,\eta)\hat{P}(\varepsilon^{'},\eta^{'})=\hat{P}(min\{\varepsilon, \varepsilon^{'}\}, min\{\eta, \eta^{'}\})$ for all $(\varepsilon,\eta)\in \Rl^2$ and $(\varepsilon^{'},\eta^{'})\in\Rl^2$;
\item[b)] $\hat{P}(\varepsilon,\eta)=0$ for all $\varepsilon<-||A||$ or $\eta<-||A||$ where $||A||$ is the Frobenius norm;
\item [c)] $\hat{P}(\varepsilon,\eta)=I$ for all $\varepsilon\geq||A||$ and $\eta\geq ||A||$;
\item [d)] $\hat{P}(\varepsilon+0,\eta+0)=\hat{P}(\varepsilon,\eta)$ for all $(\varepsilon,\eta)\in \Rl^2$
\item[e)] \be
\hat{A}=\int_{-\infty}^{\infty}\int_{-\infty}^{\infty}(\varepsilon+i \eta )d(\hat{P}_1(\varepsilon)\hat{P}_2(\eta))=\int_{-\infty}^{\infty}\int_{-\infty}^{\infty}(\varepsilon+i \eta )d(\hat{P}(\varepsilon, \eta))
\ee

\end{enumerate}
\end{Theorem}
The proof of this theorem can be found in \cite{Helmberg2008}.
\\
Given this definition of spectral decomposition for normal operators, it is now possible to define a spectral ordering. Considering two normal operators $\hat{E}=\hat{D}+i\hat{F}=\int\int \left( \alpha+i\beta \right) d(\hat{Q}_1(\alpha)\hat{Q}_2(\beta))$ and $\hat{A}=\int_{-\infty}^{\infty}\int_{-\infty}^{\infty}(\varepsilon+i \eta )d(\hat{P}(\varepsilon, \eta))$, we then define the spectral order as follows:
\be\label{equ:order1}
\hat{A}\geq_s\hat{E}\;\;\;\text{iff}\;\;\;\hat{P}_1(\varepsilon)\hat{P}_2(\eta)\leq\hat{Q}_1(\varepsilon)\hat{Q}_2(\eta)\text{ for all } (\varepsilon,  \eta)\in \Rl^2
\ee
If we consider the subspaces $\mathcal{M}_{\hat{P}}$ of the Hilbert space on which each of the individual projection operators project, the above condition is equivalent to
\be\label{equ:order2}
\hat{A}\geq_s\hat{E}\;\;\;\text{iff}\;\;\; \mathcal{M}_{\hat{P}_1(\varepsilon)}\cap \mathcal{M}_{\hat{P}_2(\eta)}\subseteq\mathcal{M}_{\hat{Q}_1(\alpha)}\cap \mathcal{M}_{\hat{Q}_2(\beta)}
\ee

However, property a) implies that for any two points $(\varepsilon,\eta)\in \Rl^2$ and $(\varepsilon^{'},\eta^{'})\in\Rl^2$ then
\be
\hat{P}(\varepsilon,\eta)\hat{P}(\varepsilon^{'},\eta^{'})=\hat{P}(\varepsilon,\eta)\text{ for }\varepsilon\leq \varepsilon^{'}\text{ and }\eta\leq \eta^{'}
\ee
Therefore, we could define
\be
\hat{P}(\varepsilon,\eta)\leq \hat{P}(\varepsilon^{'},\eta^{'})\text{ for }\varepsilon\leq \varepsilon^{'}\text{ and }\eta\leq \eta^{'}\text{ iff } \hat{P}(\varepsilon,\eta)\hat{P}(\varepsilon^{'},\eta^{'})=\hat{P}(\varepsilon,\eta)
\ee
The above reasoning shows that the spectral ordering of normal operators is intimately connected to the ordering of the self-adjoint components. 
\subsection{Ordering for the Complex Numbers}
For the case of self-adjoint operators, the spectral ordering implies an ordering of the respective spectra as follows:
\be
\hat{A}\geq_s\hat{B}\Rightarrow\lambda(\hat{A})\geq\lambda(\hat{B})
\ee
where the ordering on the right hand side is defined in $\Rl$. We would like to obtain a similar relation for normal operators, where the spectra now take its values in the complex numbers.
This is because such an ordering of the spectrum is needed when eventually defining normal operators as arrows from the state space to the topos analogue of the complex numbers (yet to be defined).

Thus we would like to define an ordering for the complex numbers compatible with the spectral ordering of normal operators. We know that $\Cl=\Rl+i\Rl$ so, in principle, we could define a partial order in terms of the order in $\Rl$ as follows:
\be
\label{complexordering}
a+ib\leq a_1+ib_1\;\;\text{iff}\;\; a\leq a_1\text{ and } b\leq b_1
\ee
However such an ordering, as will be clear later on, turns to be restrictive and incomplete. To obtain an adequate ordering for the complex numbers we first of all need to analyse whether an ordering on the spectra of normal operators is possible. 

We will perform two alternative analysis: one with respect to the projection operators in the spectral decomposition of normal operators, thus considering equation (\ref{equ:order1}) and, the other, in terms of the respective eigen subspaces, i.e. with ordering given by (\ref{equ:order2}).
We will see that the two analysis lead to the same definition.

\subsubsection{First Analysis}

We are now interested in understanding how the spectral order of normal operators is related to the order in their respective spectra.
In particular let us consider the self-adjoint operators related to $\hat{A}=\hat{C}+i\hat{B}$ and $\hat{E}=\hat{D}+i\hat{F}$, namely $\hat{A}^{'}=\hat{C}+\hat{B}$ and $\hat{E}^{'}=\hat{D}+\hat{F}$.
Moreover to really mimic the situation of the normal case we also assume that: $\hat{C}\hat{B}=\hat{B}\hat{C}$ and $\hat{D}\hat{F}=\hat{F}\hat{D}$\footnote{Since we are utilising the self-adjoint operators which comprise the normal operators $\hat{A}$ and $\hat{B}$.}.

We then obtain that if $\hat{A}^{'}\geq_s\hat{E}^{'}$ (where the ordering is now defined for self-adjoint operators, (see Appendix)) then $\lambda(\hat{A}^{'})\geq\lambda(\hat{E}^{'})$ (i.e. $\lambda(\hat{C}+\hat{B})\geq\lambda(\hat{D}+\hat{F})$).
Moreover if $\hat{A}^{'}\geq_s\hat{E}^{'}$ then $\hat{P}_1(\varepsilon)\hat{P}_2(\eta)\leq\hat{Q}_1(\varepsilon)\hat{Q}_2\eta)$ for all $(\varepsilon, \eta)\in \Rl^2$ (and vice versa), which is the same exact situation as for the normal operators.
We conclude the following set of implications:
\begin{enumerate}
\item If $\hat{A}\geq_s\hat{E}$ for normal operators, then $\hat{P}_1(\varepsilon)\hat{P}_2(\eta)\leq\hat{Q}_1(\varepsilon)\hat{Q}_2(\eta)$ for all $(\varepsilon, \eta)\in \Rl^2$, which implies that $\hat{A}^{'}\geq_s\hat{E}^{'}$ for the respective self-adjoint operators.
\item If  $\hat{A}^{'}\geq_s\hat{E}^{'}$ for self-adjoint operators, then $\hat{P}_1(\varepsilon)\hat{P}_2(\eta)\leq\hat{Q}_1(\varepsilon)\hat{Q}_2(\eta)$ for all $(\varepsilon, \eta)\in \Rl^2$,  which implies that $\hat{A}\geq_s\hat{E}$ for normal operators.
\end{enumerate}

However, since $\hat{A}^{'}\geq_s\hat{E}^{'}$ implies that $\lambda(\hat{A}^{'})\geq\lambda(\hat{E}^{'})$, i.e. $\lambda(\hat{C}+\hat{B})\geq\lambda(\hat{D}+\hat{F})$, we can then define the following ordering for the spectrum of normal operators:
\be
\lambda(\hat{C}+i\hat{B})\geq \lambda(\hat{D}+i\hat{F}) \text{ if } \lambda(\hat{C}+\hat{B})\geq\lambda(\hat{D}+\hat{F})
\ee

We thus obtain that if $\hat{A}\geq_s\hat{E}$, then $\lambda(\hat{C}+i\hat{B})\geq \lambda(\hat{D}+i\hat{F})$ in terms of the ordering given above\footnote{ In detail, we have that $\hat{A}\geq_s\hat{E}$ (for normal operators) implies $\hat{P}_1(\varepsilon)\hat{P}_2(\eta)\leq\hat{Q}_1(\varepsilon)\hat{Q}_2(\eta)$ which, in turn, implies that $\hat{A}^{'}\geq_s\hat{E}^{'}$ for the respective self-adjoint operators. As a consequence $ \lambda(\hat{C}+\hat{B})\geq\lambda(\hat{D}+\hat{F})$, which, from the above definition implies that $\lambda(\hat{C}+i\hat{B})\geq \lambda(\hat{D}+i\hat{F})$. It follows that we can now state that if $\hat{A}\geq_s\hat{E}$, then $\lambda(\hat{A})\geq\lambda(\hat{B})$ ($ \lambda(\hat{C}+i\hat{B})\geq\lambda(\hat{D}+i\hat{F})$).}.
\\
\begin{Definition}
Given two normal operators $ \hat{A}= \hat{C}+i\hat{B}$ and $\hat{E}=\hat{D}+i\hat{F}$, if $\hat{A}\geq_s \hat{E}$ with respect to the spectral order of normal operators defined in (\ref{equ:order1}) and (\ref{equ:order2}), then $\lambda(\hat{C}+i\hat{B})\geq \lambda(\hat{D}+i\hat{F}) $.
\end{Definition}
\subsubsection{Second Analysis}

Let us assume that we have two normal operators $\hat{C}=\hat{A}+i\hat{B}$ and $\hat{C}_1=\hat{A}_1+i\hat{B}_1$, such that $\hat{C}\geq_s\hat{C}_1$ which, according to the spectral theorem, implies that $\hat{E}^{\hat{C}}_{\lambda}\leq\hat{E}^{\hat{C}_1}_{\lambda}$.
Therefore, for each $\lambda\in \Cl$,  the vector space spanned by the eigenvectors of $\hat{E}_{\lambda}^{\hat{C}}$ is a subspace of the space spanned by the eigenvectors of $\hat{E}_{\lambda}^{\hat{C}_1}$. Now, since $[\hat{A},\hat{B}]=0$ and
$[\hat{A}_1,\hat{B}_1]=0$, then each of the commuting pairs has a common set of eigenvectors.
Let us take an eigenvector $|\psi\rangle$ common to both $\hat{A}$ and $\hat{B}$, which is obviously also an eigenvector of $\hat{C}$.
We then have $\left( \hat{A}+i\hat{B} \right) (|\psi\rangle)=\hat{A}|\psi\rangle+i\hat{B} |\psi\rangle= (\gamma+i\sigma)$.

Thus the question is: what is the relation between the space of eigenvectors of $\hat{A}$ and that of $\hat{B}$?
To this end let us simplify the situation and consider the sum $\hat{D}=\hat{A}+\hat{B}$. 
The space of eigenvectors of $\hat{A}+\hat{B}$  will certainly be smaller than the space of eigenvectors of either $\hat{A}$ or $\hat{B}$.
It will actually be the intersection of the space of eigenvectors of $\hat{A}$ and $\hat{B}$, since $(\hat{A}+\hat{B})|\psi\rangle=\hat{A}|\psi\rangle+\hat{B}|\psi\rangle$.
It follows that $\mathcal{M}_{\hat{E}^{\hat{D}}_{\lambda}}=\mathcal{M}_{\hat{E}^{\hat{A}}_{\gamma}}\cap \mathcal{M}_{\hat{E}^{\hat{B}}_{\sigma}}$ \footnote{This is equivalent to $\hat{E}^{\hat{D}}_{\lambda}=\hat{E}^{\hat{A}}_{\gamma}\wedge\hat{E}^{\hat{B}}_{\sigma} $, where $\hat{P}\hat{Q}=\hat{P}\wedge\hat{Q}$ for any two projection operators.}, for all $\lambda=\gamma+\sigma\in \Rl$.

Given another operator $\hat{D}_1=\hat{A}_1+\hat{B}_1$, such that $\hat{D}\geq_s\hat{D}_1$, the definition of the spectral ordering for self-adjoint operators implies that $\mathcal{M}_{\hat{E}^{\hat{A}}_{\gamma}}\cap\mathcal{M}_{\hat{E}^{\hat{B}}_{\sigma}}\subseteq \mathcal{M}_{\hat{E}^{\hat{A}_1}_{\gamma}}\cap\mathcal{M}_{\hat{E}^{\hat{B}_1}_{\sigma}} $. It follows that $ \hat{A}+\hat{B}\geq_s \hat{A}_1+\hat{B}_1$ and consequently $\lambda(\hat{A}+\hat{B})\geq\lambda(\hat{A}_1+\hat{B}_1)$.\\
We now go back to considering normal operators.
We assume that $\hat{C}\geq_s\hat{C}_1$, which from the definition of the spectral ordering of normal operators implies $\hat{E}^{\hat{C}}_{\lambda}\leq\hat{E}^{\hat{C}_1}_{\lambda}$. Since each eigenvalue of both $\hat{C}$ and $\hat{C}_1$ will be of the form $\lambda=\gamma+i\sigma$ and $\lambda_1=\gamma_1+i\sigma_1$, respectively, it is possible to uniquely define the following associations:
\ba
\lambda&\rightarrow&\mathcal{M}_{ \hat{E}^{\hat{A}}_{\gamma}}\cap \mathcal{M}_{\hat{E}^{\hat{B}}_{\sigma}}\\
\lambda_1&\rightarrow& \mathcal{M}_{\hat{E}^{\hat{A}}_{\gamma_1}}\cap\mathcal{M}_{ \hat{E}^{\hat{B}}_{\sigma_1}}
\ea
for all $\lambda\in sp(\hat{C})$ and $\lambda_1\in sp(\hat{C}_1)$ (this can obviously be extended to all the complex numbers). This means that each eigenvector of $\hat{C}$ will be isomorphic to one contained in the subspace spanned by $\hat{E}^{\hat{A}}\wedge \hat{E}^{\hat{B}}$ and similarly for the eigenvectors of $\hat{C}_1$.
As a result the subspaces (of the Hilbert space) $\mathcal{M}_{\hat{E}^{\hat{C}}_{\lambda}}$ and $\mathcal{M}_{\hat{E}^{\hat{A}}_{\gamma}}\cap \mathcal{M}_{\hat{E}^{\hat{B}}_{\sigma}}$, for all $\lambda=\gamma+i\sigma$ are isomorphic.
Therefore, if $\hat{C}\geq_s\hat{C}_1$ such that $\mathcal{M}_{\hat{E}^{\hat{C}}_{\lambda}}\subseteq\mathcal{M}_{\hat{E}^{\hat{C}_1}_{\lambda}}$, then $\mathcal{M}_{\hat{E}^{\hat{A}}_{\gamma}}\cap \mathcal{M}_{\hat{E}^{\hat{B}}_{\sigma}}\subseteq\mathcal{M}_{\hat{E}^{\hat{A}_1}_{\gamma_1}}\cap \mathcal{M}_{\hat{E}^{\hat{B_1}}_{\sigma_1}}$ and vice versa\footnote{All the above might be a direct consequence of the isomorphisms $\Cl\simeq \Rl\times\Rl$.}.
\\
However, since $\mathcal{M}_{\hat{E}^{\hat{A}}_{\gamma}}\cap \mathcal{M}_{\hat{E}^{\hat{B}}_{\sigma}}\subseteq\mathcal{M}_{\hat{E}^{\hat{A}_1}_{\gamma_1}}\cap \mathcal{M}_{\hat{E}^{\hat{B_1}}_{\sigma_1}}$ implies that $\hat{A}+\hat{B}\geq\hat{A}_1+\hat{B}_1$, it follows that $\lambda(\hat{A}+\hat{B})\geq\lambda(\hat{A}_1+\hat{B}_1)$.

What the above reasoning reveals is that it is possible to define an ordering on the spectrum of normal operators even if it consists of complex numbers. In particular, we can now define the following:
\be
\lambda(\hat{A}+i\hat{B})\geq \lambda(\hat{A}_1+i\hat{B}_1) \text{ if } \lambda(\hat{A}+\hat{B})\geq\lambda(\hat{A}_1+\hat{B}_1)
\ee
Since each normal operator is defined as a complex sum of self-adjoint operators, the ordering is well defined for all normal operators.

Therefore, given two normal operators $\hat{C}$ and $\hat{C}_1$, then $\hat{C}\geq_s\hat{C}_1$ iff $\lambda(\hat{A}+\hat{B})\geq\lambda(\hat{A}_1+\hat{B}_1)$. 
We can make two statements:
\begin{itemize}
\item [i)] If $\hat{C}\geq_s\hat{C}_1$ then $\mathcal{M}_{\hat{E}^{\hat{A}}_{\gamma}}\cap \mathcal{M}_{\hat{E}^{\hat{B}}_{\sigma}}\subseteq\mathcal{M}_{\hat{E}^{\hat{A}_1}_{\gamma_1}}\cap \mathcal{M}_{\hat{E}^{\hat{B_1}}_{\sigma_1}}$, therefore $\hat{A}+\hat{B}\geq\hat{A}_1+\hat{B}_1$ which implies $\lambda(\hat{A}+\hat{B})\geq\lambda(\hat{A}_1+\hat{B}_1)$.

\item[ii)] If $\lambda(\hat{A}+\hat{B})\geq\lambda(\hat{A}_1+\hat{B}_1)$, then $\hat{A}+\hat{B}\geq \hat{A}_1+\hat{B}_1$ which implies $\mathcal{M}_{\hat{E}^{\hat{A}}_{\gamma}}\cap \mathcal{M}_{\hat{E}^{\hat{B}}_{\sigma}}\subseteq\mathcal{M}_{\hat{E}^{\hat{A}_1}_{\gamma_1}}\cap \mathcal{M}_{\hat{E}^{\hat{B_1}}_{\sigma_1}}$, therefore $\hat{C}\geq_s\hat{C}_1$.

\end{itemize}
It follows that if $\hat{C}\geq_s\hat{C}_1$, then $\lambda(\hat{A}+i\hat{B})\geq \lambda(\hat{A}_1+i\hat{B}_1)$. 

Given the results of our two analyses we can now define an ordering for the complex numbers as follows:

\begin{Definition}\label{def:order}
Given two complex numbers $\lambda_1=\epsilon_1+i\eta_1$ and $\lambda=\epsilon+i\eta$, then we say that
\be
\lambda_1\geq\lambda\text{ if } (\epsilon_1+\eta_1)\geq(\epsilon+\eta)
\ee
Where $(\epsilon_1+\eta_1)\geq(\epsilon+\eta)$ obeys the usual ordering of the reals.

\end{Definition}

\subsection{An Example}
\label{example1}

An example of the operator ordering we have defined can be illustrated by two non-self-adjoint bounded operators
with finite spectra.
Let us consider a two-state system with non-self-adjoint operator
\begin{equation}
\hat{O}_z = \left(
\begin{array}{c c}
1 &  0 \\
0 & -i
\end{array}
\right)
\end{equation}
and its norm squared
\begin{equation}
\hat{O}^2_z = \left(
\begin{array}{c c}
1 & 0 \\
0 & 1
\end{array}
\right)
\end{equation}

We can decompose $\hat{O}_z$ into two matrices with eigenprojectors $\hat{P}_1, \hat{P}_2$ and $\hat{Q_1}, \hat{Q}_2$:
\begin{equation}
\hat{O}_z = \hat{A} + i \hat{B} = \left(
\begin{array}{c c}
1 & 0 \\
0 & 0
\end{array}
\right)
+ i \left(
\begin{array}{c c}
0 & 0 \\
0 & -1
\end{array}
\right)
\end{equation}

The spectral family for $\lambda = \epsilon + i \eta$ is then as follows:
\begin{equation}
\label{specfam}
\hat{E}_{\lambda}^{\hat{O}_z} = \left\{
	\begin{array}{l l}
	\hat{0} \hat{0} & \text{if } \epsilon < 0, \eta < -1 \\
	\hat{0} \hat{Q}_2 & \text{if } \epsilon < 0, -1 \le \eta < 0 \\
	\hat{0} \left( \hat{Q}_1 + \hat{Q}_2 \right) & \text{if } \epsilon < 0, 0 \le \eta \\
	\hat{P}_2 \hat{0} & \text{if } 0 \le \epsilon < 1, \eta < -1 \\
	\hat{P}_2 \hat{Q}_2 & \text{if } 0 \le \epsilon < 1, -1 \le \eta < 0 \\
	\hat{P}_2 \left( \hat{Q}_1 + \hat{Q}_2 \right) & \text{if } 0 \le \epsilon < 1, 0 \le \eta \\
	\left( \hat{P}_1 + \hat{P}_2 \right) \hat{0} & \text{if } 1 \le \epsilon, \eta < -1 \\
	\left( \hat{P}_1 + \hat{P}_2 \right) \hat{Q}_2 & \text{if } 1 \le \epsilon, -1 \le \eta < 0 \\
	\left( \hat{P}_1 + \hat{P}_2 \right) \left( \hat{Q}_1 + \hat{Q}_2 \right) & \text{if } 1 \le \epsilon, 0 \le \eta \\
	\end{array}
	\right.
\end{equation}

and

\begin{equation*}
\hat{E}_{\lambda}^{\hat{O}_z^2} = \left\{
        \begin{array}{l l}
        \hat{0} \hat{0} & \text{if } \epsilon < 1, \eta < 0 \\
        \hat{0} \hat{Q} & \text{if } \epsilon < 1, 0 \le \eta \\
        \left( \hat{P}_1 + \hat{P}_2 \right) \hat{0} & \text{if } \epsilon \ge 1, \eta < 0 \\
        \left( \hat{P}_1 + \hat{P}_2 \right) \hat{Q} & \text{if } \epsilon \ge 1, 0 \le \eta \\
        \end{array}
        \right.
\end{equation*}
where $\hat{Q} \equiv \hat{Q}_1 + \hat{Q}_2$.
By comparing the spectral families in a piecewise manner,
one can see that for any breakdown of $\epsilon,\eta$, we have
that $ \hat{E}_{\lambda}^{\hat{O}_z} \ge  \hat{E}_{\lambda}^{\hat{O}_z^2} $
and so $\hat{O}_z \le_{S} \hat{O}_z^2$.

By the results shown above, in order to be able to compare these operators, we can also
write the spectral decomposition of the sum of the real and imaginary operator parts.
For the first this becomes
\begin{equation}
\hat{A} + \hat{B} = \left(
\begin{array}{c c}
1 &  0 \\
0 & -1
\end{array}
\right)
\end{equation}
while the second stays the same.
The spectral decomposition of these operators is then
\begin{equation*}
\hat{E}_{\lambda}^{\hat{A}+\hat{B}} = \left\{
        \begin{array}{l l}
        \hat{0} & \text{if } \lambda < -1 \\
        \hat{R}_2 & \text{if } -1 \le \lambda < 1 \\
        \left( \hat{R}_1 + \hat{R}_2 \right) & \text{if } 1 \le \lambda \\
        \end{array}
        \right.
\end{equation*}
and
\begin{equation*}
\hat{E}_{\lambda}^{|\hat{O}_z^2|} = \left\{
        \begin{array}{l l}
        \hat{0} & \text{if } \lambda < 1 \\
        \left( \hat{R}_1 + \hat{R}_2 \right) & \text{if } \lambda \ge 1 \\
        \end{array}
        \right.
\end{equation*}
where $\hat{R}_1,\hat{R}_2$ are the projectors in the non-complex space.

We can see that the natural spectral ordering implies
\begin{equation*}
\left\{
	\begin{array}{c c}
	\hat{E}_{\lambda}^{|\hat{O}_z^2|} = \hat{E}_{\lambda}^{\hat{A}+\hat{B}} & \text{if } \lambda < -1 \\
	\hat{E}_{\lambda}^{|\hat{O}_z^2|} \le \hat{E}_{\lambda}^{\hat{A}+\hat{B}} & \text{if } -1 \le \lambda < 1 \\
	\hat{E}_{\lambda}^{|\hat{O}_z^2|} = \hat{E}_{\lambda}^{\hat{A}+\hat{B}} & \text{if } 1 \le \lambda \\
	\end{array}
\right.
\end{equation*}
which implies that ${\hat{A}+\hat{B}} \le_{s} {|\hat{O}_z^2|}$ and, by definition (\ref{def:order}),
$\hat{O}_z \le_{s} \hat{O}_z^2$.
Therefore, both treatments are equivalent.

\section{Daseinisation of Normal Operators}

In this section we try to extend the daseinisation of self-adjoint operators
to the daseinisation of normal operators. To this end we need to extend the concept of the Gel'fand Transform to normal operators.

\subsection{The Gel'fand Transform}

The Gel'fand representation theorem states that for the Gel'fand spectrum $\Sigma_V$ of a self-adjoint operator $\hat{A}$,
there exists an isomorphism given by \cite{Doering2008}
\begin{align*}
V &\rightarrow C(\Sigma_V) \\
\hat{A} &\rightarrow \bar{A} \equiv \lambda(\hat{A})
\end{align*}
where $\bar{A}$ is also denoted as the {\it Gel'fand transform} of the self-adjoint operator $\hat{A}$.

Firstly, do we have a Gel'fand representation theorem for normal operators?
Indeed, we do \cite{Whitley1968}.
For the closed $^*$-sub-algebra generated by a normal operator $T$, $T^{*}$, and the identity element,
there exists a mapping onto the space of $\Sigma(T)$ given by
\begin{align*}
\mathcal{A} &\rightarrow C(\Sigma(T)) \\
\hat{T} &\rightarrow \bar{T} \equiv \lambda_C(\hat{T})
\end{align*}
where $ C(\Sigma(T))$ is the space of complex continuous functions on $\Sigma(T)$.

We therefore intend to define the inner and outer daseinisations in the same manner as for self-adjoint operators, namely in terms of the Gel'fand transforms:
\ba
\overline{\delta^o(\hat{C})}_V:\us_V&\rightarrow&\Cl\\
\overline{\delta^i(\hat{C})}_V:\us_V&\rightarrow&\Cl
\ea
However, in the case of self-adjoint operator we know that for a sub-context $V^{'}\subseteq V$, then $\delta^o(\hat{A})_{V^{'}}\geq\delta^o(\hat{A})_{V}$ and $\delta^i(\hat{A})_{V^{'}}\leq\delta^i(\hat{A})_{V}$. These relations imply that the respective Gel'fand transforms undergo the following relations: $\overline{\delta^o(\hat{A})_{V^{'}}}(\lambda_{|V^{'}})\geq\overline{\delta^o(\hat{A})_{V}}(\lambda)$ and $\overline{\delta^i(\hat{A})_{V^{'}}}(\lambda_{|V^{'}})\leq\overline{\delta^i(\hat{A})_{V}}(\lambda)$. We want similar relations to hold for normal operators.

\subsection{Daseinisation}
We know that for each self-adjoint operator $\hat{A}$ we have a spectral family $\{\hat{E}^{\hat{A}}_{\lambda} | \forall \lambda\in\sigma(\hat{A})\}$ which can be extended to all $\lambda\in \Rl$ (see Appendix for details).
Moreover, it was shown by de Groote 
(\cite{Groote2007b}, \cite{Groote2007a}),
that if $\lambda\rightarrow  \hat{E}_{\lambda}$ is a spectral family in $P(\mh)$ (or, equivalently, a self-adjoint operator $\hat{A}$),
then, for each context $V\in\mv(\mh)$ , the maps
\ba
\lambda&\rightarrow& \delta^i(\hat{E}_{\lambda})_V\\
\lambda&\rightarrow&\bigwedge_{\mu>\lambda} \delta^o(\hat{E}_{\mu})_V
\ea
are also spectral families. Since these spectral families lie in $P(V)$ they define
self-adjoint operators in $V$.

Similarly, for a normal operator $\hat{A}$ in $\mb(\mh)$, then there exists a unique spectral family $\{\hat{P}(\varepsilon, \eta):=\hat{P}_1(\varepsilon)\hat{P}_2(\eta)|(\varepsilon, \eta)\in\Rl^2\}$ of projection operators. 
Thus, applying the exact same proof as was used in \cite{Groote2007b}, it follows that the following are themselves spectral families
\ba
\lambda=\varepsilon+i \eta&\rightarrow& \delta^i(\hat{P}(\varepsilon, \eta))_V\\
\lambda=\varepsilon+i \eta&\rightarrow& \bigwedge_{\mu>\lambda} \delta^o(\hat{P}(\varepsilon^{'}, \eta^{'}))_V
\ea
where $\mu=\varepsilon^{'}+i\eta^{'}$, and the ordering $\mu>\lambda$ is given in definition \ref{def:order}.
Using this, we can define the daseinisation of normal operators as follows:
\begin{Definition}
Let $\hat{A}=\hat{C}+i\hat{D}$ be an arbitrary normal operator. Then the outer and inner daseinisations of $\hat{A}$ are defined in each sub-context $V$ as:
\ba
\delta^o(\hat{A})_V&:=&\int_{\Cl}\lambda d(\delta^i(\hat{E}_{\lambda})_V=\int_{\Cl}(\varepsilon+i\eta)d(\delta^i(\hat{P}(\varepsilon, \eta)_V))\\
\delta^i(\hat{A})_V&:=&\int_{\Cl}\lambda d \left(\bigwedge_{\mu>\lambda}\delta^o(\hat{E}_{\mu})_V\right) = \int_{\Cl}(\varepsilon+i\eta)d \left(\bigwedge_{\mu>\lambda}\delta^o(\hat{P}(\varepsilon^{'},\eta)_V) \right)
\ea
respectively.
\end{Definition}

Since for all $V\in \mv(\mh)$
\ba
\delta^i(\hat{P}\wedge\hat{Q})_V&\leq& \delta^i(\hat{P})_V\wedge\delta^i(\hat{Q})_V\\
\delta^o(\hat{P}\wedge\hat{Q})_V&\geq& \delta^o(\hat{P})_V\wedge\delta^o(\hat{Q})_V
\ea
it follows that
\ba
\delta^o(\hat{A})_V&\geq &\delta^o(\hat{C})_V+i\delta^o(\hat{D})_V\\
\delta^i(\hat{A})_V&\leq&\delta^i(\hat{C})_V+i\delta^i(\hat{D})_V
\ea

Moreover, from the definition of inner and outer daseinisation of projection operators, for all $V\in\mv(\mh)$ and all $\lambda\in\Cl$ we have
\be
\bigwedge_{\mu\geq\lambda}\delta^o(\hat{E}(\mu))_V\geq\delta^i(\hat{E}(\lambda))_V
\ee

Therefore, from the definition of spectral order it follows that
\be
\delta^i(\hat{A})_V\leq_s\delta^o(\hat{A})_V
\ee
We would now like to analyse the spectrum of these operators. 
As a consequence of the spectral theorem and the fact that both $\delta^i(\hat{A})_V$ and $\delta^o(\hat{A})_V$ are in $V$, it is possible to represent them through the Gel'fand transform as follows:
\ba
\overline{\delta^o(\hat{A})_V}:\us_V&\rightarrow&sp(\delta^o(\hat{A})_V)\\
\overline{\delta^i(\hat{A})_V}:\us_V&\rightarrow&sp(\delta^i(\hat{A})_V)
\ea

Since $sp(\delta^o(\hat{A})_V)\subseteq \Cl$ and $sp(\delta^i(\hat{A})_V)\subseteq\Cl$ we can generalise the above maps to
\ba
\overline{\delta^o(\hat{A})_V}:\us_V&\rightarrow&\Cl\\
\overline{\delta^i(\hat{A})_V}:\us_V&\rightarrow&\Cl
\ea
However, the relation $\delta^i(\hat{A})_V\leq_s\delta^o(\hat{A})_V$ together with the spectral ordering implies that 
for all $V\in\mv(\mh)$,  $\overline{\delta^i(\hat{A})_V}(\lambda)\leq\overline{\delta^o(\hat{A})_V}(\lambda)$ (where again the ordering is the one defined in \ref{def:order}).

Moreover, as we go to smaller sub-algebras $V^{'}\subseteq V$, since $ \delta^o(\hat{E}(\lambda))_{V^{'}}\geq\delta^o(\hat{E}(\lambda))_V$ and $ \delta^i(\hat{E}(\lambda))_{V^{'}}\leq  \delta^i(\hat{E}(\lambda))_{V}$, then $\delta^i(\hat{A})_V\geq_s\delta^i(\hat{A})_{V^{'}}$ while $\delta^o(\hat{A})_V\leq_s\delta^o(\hat{A})_{V^{'}}$. 
Thus, inner daseinisation preserves the order while outer daseinisation reverses the order. 
As a consequence we obtain the following:
\ba
\overline{\delta^o(\hat{A})_V}(\lambda)&\leq&\overline{\delta^o(\hat{A})_{V^{'}}}(\lambda_{|V^{'}})\\
\overline{\delta^i(\hat{A})_V}(\lambda)&\geq&\overline{\delta^i(\hat{A})_{V^{'}}}(\lambda_{|V^{'}})
\ea
where $\overline{\delta^o(\hat{A})_V}(\lambda):=\lambda(\delta^o(\hat{A})_V)$ while $\overline{\delta^o(\hat{A})_{V^{'}}}\lambda_{|V^{'}}:=\lambda_{|V^{'}}(\delta^o(\hat{A})_{V^{'}})$.

\subsection{Daseinisation of our Example State}

Building on to our example in section \ref{example1}, we can explore the
daseinisation of the operator
\begin{equation}
\hat{O}_z = \left(
\begin{array}{c c}
1 &  0 \\
0 & -i
\end{array}
\right)
= \left(
\begin{array}{c c}
1 & 0 \\
0 & 0
\end{array}
\right)
+ i \left(
\begin{array}{c c}
0 & 0 \\
0 & -1
\end{array}
\right)
\end{equation}

We have four projectors, corresponding to the projectors
$\hat{P}_1,\hat{P}_2,\hat{Q}_1,\hat{Q}_2$ above, along with
the projectors $\hat{0}$ and $\hat{1}$.

Therefore, we can use the daseinisation of the spectral family
of our operator (\ref{specfam}) to define the daseinisation;
our problem is simply reduced to the daseinisations
$\delta^0(\hat{P}_1)_V,\delta^0(\hat{P}_2)_V,\delta^0(\hat{Q}_1)_V, \delta^0(\hat{Q}_2)_V$.

Let's choose a sub-context $V$ spanned by the projection operators
$\hat{P}_1+\hat{P}_2,\hat{Q}_1,\hat{Q}_2$.
Then the only nontrivial outer daseinisations are
$\delta^0(\hat{P}_1)_V = \delta^0(\hat{P}_2)_V = \hat{P}_1+\hat{P}_2$.

The spectral family of our daseinised operator is
\begin{equation}
\label{specfam2}
\delta^0 \left( \hat{E}_{\lambda}^{\hat{O}_z} \right)_V = \left\{
        \begin{array}{l l}
        \hat{0} \hat{0} & \text{if } \epsilon < 0, \eta < -1 \\
        \hat{0} \hat{Q}_2 & \text{if } \epsilon < 0, -1 \le \eta < 0 \\
        \hat{0} \left( \hat{Q}_1 + \hat{Q}_2 \right) & \text{if } \epsilon < 0, 0 \le \eta \\
        \left( \hat{P}_1 + \hat{P}_2 \right) \hat{0} & \text{if } 0 \le \epsilon, \eta < -1 \\
        \left( \hat{P}_1 + \hat{P}_2 \right) \hat{Q}_2 & \text{if } 0 \le \epsilon, -1 \le \eta < 0 \\
        \left( \hat{P}_1 + \hat{P}_2 \right) \left( \hat{Q}_1 + \hat{Q}_2 \right) & \text{if } 0 \le \epsilon, 0 \le \eta \\
        \end{array}
        \right.
\end{equation}

Recall that the spectral family of the operator $\hat{E}_{\lambda}^{\hat{O}_z}$ was
\begin{equation*}
\hat{E}_{\lambda}^{\hat{O}_z} = \left\{
        \begin{array}{l l}
        \hat{0} \hat{0} & \text{if } \epsilon < 0, \eta < -1 \\
        \hat{0} \hat{Q}_2 & \text{if } \epsilon < 0, -1 \le \eta < 0 \\
        \hat{0} \left( \hat{Q}_1 + \hat{Q}_2 \right) & \text{if } \epsilon < 0, 0 \le \eta \\
        \hat{P}_2 \hat{0} & \text{if } 0 \le \epsilon < 1, \eta < -1 \\
        \hat{P}_2 \hat{Q}_2 & \text{if } 0 \le \epsilon < 1, -1 \le \eta < 0 \\
        \hat{P}_2 \left( \hat{Q}_1 + \hat{Q}_2 \right) & \text{if } 0 \le \epsilon < 1, 0 \le \eta \\
        \left( \hat{P}_1 + \hat{P}_2 \right) \hat{0} & \text{if } 1 \le \epsilon, \eta < -1 \\
        \left( \hat{P}_1 + \hat{P}_2 \right) \hat{Q}_2 & \text{if } 1 \le \epsilon, -1 \le \eta < 0 \\
        \left( \hat{P}_1 + \hat{P}_2 \right) \left( \hat{Q}_1 + \hat{Q}_2 \right) & \text{if } 1 \le \epsilon, 0 \le \eta \\
        \end{array}
        \right.
\end{equation*}

We can see that for any $\epsilon, \eta$, $\hat{E}_{\lambda}^{\hat{O}_z} \le  \delta^0 \left( \hat{E}_{\lambda}^{\hat{O}_z} \right)_V$ and therefore $\delta^0 \left(\hat{O}_z \right)_V \le_s \hat{O}_z$ as desired.

\section{Complex Numbers in a Topos}\label{sec:complex}
Complex numbers in a topos have been previously defined in various papers \cite{maklane}, \cite{Banasche1} and \cite{Banasche2}, however the definition of these objects did not take into account the spectra of normal operators. In the present situation, since our ultimate aim is to define normal operators as maps from the state space to the (complex) quantity value object, we have to resort to a different characterisation of complex numbers in a topos. This will be very similar to how the real quantity value object is defined.

\begin{Definition}
The complex quantity value object is the presheaf $\underline{\Cl}^{\leftrightarrow}$ which has as
\begin{itemize}
\item Objects: For all contexts $V\in\mv(\mh)$,
\be
\underline{\Cl}^{\leftrightarrow}_V:=\{(\mu, \nu)|\mu\in OP(\downarrow V, \Cl), \nu\in OR(\downarrow V, \Cl); \mu\leq \nu\}
\ee
Where $OP$ denotes the set of order preserving functions, while $OR$ the set of order reversing functions.
\item Morphisms: Given a map between contexts $i_{V^{'}V}:V^{'}\subseteq V$ the corresponding morphisms are
\ba
\underline{\Cl}^{\leftrightarrow}(i_{V^{'}V}):\underline{\Cl}^{\leftrightarrow}_V&\rightarrow& \underline{\Cl}^{\leftrightarrow}_{V^{'}}\\
(\mu, \nu)&\mapsto&(\mu_{|V^{'}}, \nu_{|V^{'}})
\ea
\end{itemize}
where $\mu_{V^{'}}:\downarrow V^{'}\rightarrow \Cl$ is simply the restriction of $\mu$ to the sub-context $V^{'}$.
\end{Definition}

This definition suits our purpose: we need to preserve the fact that under outer daseinisation we obtain the inequality
 $\overline{\delta^o(\hat{A})_V}(\lambda)\leq\overline{\delta^o(\hat{A})_{V^{'}}}\lambda_{|V^{'}}$
 and under inner daseinisation we have 
$\overline{\delta^i(\hat{A})_V}(\lambda)\geq\overline{\delta^i(\hat{A})_{V^{'}}}\lambda_{|V^{'}}$.

Furthermore, we can use the isometry
\be
\Rl\times \Rl\simeq\Cl
\ee
as a guideline to rigorously define the transformation between the quantity value object
and the complex quantity value object.
Recall that the quantity value object is a monoid (a semigroup with unit) and, as such, it is equipped with the summation operation 
\be
+:\underline{\Rl}^{\leftrightarrow}\times \underline{\Rl}^{\leftrightarrow}\rightarrow \underline{\Rl}^{\leftrightarrow}
\ee
which is defined for each $V\in \mv(\mh)$ as

\ba
+_V:\underline{\Rl}^{\leftrightarrow}_V\times \underline{\Rl}^{\leftrightarrow}_V&\rightarrow& \underline{\Rl}^{\leftrightarrow}_V\\
\Big((\mu_1,\nu_1),(\mu_2,\nu_2)\Big)&\mapsto&(\mu_1+\mu_2,\nu_1+\nu_2)=(\mu, \nu)
\ea
Here $(\mu_1+\mu_2,\nu_1+\nu_2)$ is defined, for each $V^{'}\subseteq V$ as $(\mu_1(V)+\mu_2(V),\nu_1(V)+\nu_2(V))$.
We can make use of this to define the map
\ba
f_V:\underline{\Rl}^{\leftrightarrow}_V\times \underline{\Rl}^{\leftrightarrow}_V&\rightarrow& \underline{\Cl}^{\leftrightarrow}_V\\
\Big((\mu_1,\nu_1),(\mu_2,\nu_2)\Big)&\mapsto&(\mu_1+i\mu_2, \nu_1+i\nu_2)
\ea
Even in this case, for each $V^{'}\subseteq V$ the above complex sum should be intended as
\ba
\mu (V^{'})&:=&\mu_1(V^{'})+i\mu_2(V^{'})\\
\nu (V^{'})&:=&\nu_1(V^{'})+i\nu_2(V^{'})
\ea
for each context  $V\in\mv(\mh)$.\\Thus the map $f_V$ takes the pair $((\mu_1, \nu_1), (\mu_2, \nu_2))\in\ps{R}^{\leftrightarrow}\times\ps{\Rl^{\leftrightarrow}}$ and maps it to the element $(\mu, \nu)\in\ps{\Cl^{\leftrightarrow}}$ consisting of a pair of order reversing and order preserving maps such that for each $V\in\mv(\mh)$. Such a pair is defined as follows:
$\mu (V):=\mu_1(V)+i\mu_2(V)$ and $\nu (V^{'}):=\nu_1(V)+i\nu_2(V)$.\\
Therefore, we have a relationship
\be
f:\underline{\Rl}^{\leftrightarrow}\times \underline{\Rl}^{\leftrightarrow}\rightarrow \underline{\Cl}^{\leftrightarrow}
\ee
between the quantity value object and the complex quantity value object.

First of all we will show that $f$ is indeed a natural transformation. To this end we need to show that the following diagram commutes
\[\xymatrix{
\underline{\Rl}^{\leftrightarrow}_V\times \underline{\Rl}^{\leftrightarrow}_V\ar[rr]^{f_V}\ar[dd]_g&&\underline{\Cl}^{\leftrightarrow}_V\ar[dd]^h\\
&&\\
\underline{\Rl}^{\leftrightarrow}_{V^{'}}\times \underline{\Rl}^{\leftrightarrow}_{V^{'}}\ar[rr]_{f_{V^{'}}}&&\underline{\Cl}^{\leftrightarrow}_{V^{'}}\\
}\]
where $h$ and $g$ are the presheaf maps, i.e. we want to show that $h\circ f_V=f_{V^{'}}\circ g$. Let us consider an element $\Big((\mu_1,\nu_1),(\mu_2,\nu_2)\Big)$. Chasing the diagram around in one direction we have
\be
h\circ f_V\Big((\mu_1,\nu_1),(\mu_2,\nu_2)\Big)=h(\mu_1+i\mu_2, \nu_1+i\nu_2)=\Big((\mu_1+i\mu_2)_{|V^{'}}, (\nu_1+i\nu_2)_{|V^{'}}\Big)
\ee
where $(\mu_1+i\mu_2)_{|V^{'}}=(\mu_1)_{|V^{'}}+i(\mu_2)_{|V^{'}}$.
On the other hand
\be
f_{V^{'}}\circ g\Big((\mu_1,\nu_1),(\mu_2,\nu_2)\Big)=f_{V^{'}}\Big((\mu_1,\nu_1)_{|V^{'}},(\mu_2,\nu_2)_{|V^{'}}\Big)=\Big((\mu_1)_{|V^{'}}+i(\mu_2)_{|V^{'}},(\nu_1)_{|V^{'}}+i(\nu_2)_{|V^{'}}\Big)
\ee
Thus indeed $f$ is a natural transformation.\\

It can also be shown that $f$ is 1:1. 
If $f_V\Big((\mu_1,\nu_1),(\mu_2,\nu_2)\Big)=f_V\Big((\mu^{'}_1,\nu^{'}_1),(\mu^{'}_2,\nu^{'}_2)\Big)$ then $(\mu_1+i\mu_2, \nu_1+i\nu_2)=(\mu^{'}_1+i\mu^{'}_2, \nu^{'}_1+i\nu^{'}_2)$. Therefore $(\mu_1+i\mu_2)=(\mu^{'}_1+i\mu^{'}_2)$ and $(\nu_1+i\nu_2)=(\nu^{'}_1+i\nu^{'}_2)$.  By evaluating such maps at each $V$ it follows that
$\Big((\mu_1,\nu_1),(\mu_2,\nu_2)\Big)=\Big((\mu^{'}_1,\nu^{'}_1),(\mu^{'}_2,\nu^{'}_2)\Big)$.
However, the map $f$ is not onto. 
This is because the ordering that we defined on $\underline{\Cl}^{\leftrightarrow}$ is more general than the ordering coming from pairs of order reversing and order preserving maps.
In fact, for $\mu_1+i\mu_2\leq\mu_3+i\mu_4$ we only require that
$\mu_1+\mu_2\leq\mu_3+\mu_4$, not that $\mu_1\leq \mu_3$ and $\mu_2\leq \mu_4$.
Obviously, the latter relation implies the former, but the converse is not true. Thus it follows that $\ps{\Rl^{\leftrightarrow}}\times\ps{\Rl^{\leftrightarrow}}$ will be isomorphic to a sub-object of $\ps{\Cl^{\leftrightarrow}}$.

\subsection{Properties of $\underline{\Cl}^{\leftrightarrow}$}
Given the object $\Cl^{\leftrightarrow}$ defined above we are interested in analysing what types of properties it has. In particular, we know that $\Cl$ is a group and it is also a vector space over the reals. Can the same be said for $\ps{\Cl^{\leftrightarrow}}$? We first analyse whether the usual operations present in $\Cl $ are also present in $\ps{ \Cl}^{\leftrightarrow}$.
\begin{enumerate}
\item {\it Conjugation}.
The most obvious way of defining conjugation would be the following: \\
For each $V\in \mv(\mh)$ we have
\ba
^*_V:\ps{\Cl^{\leftrightarrow}}_V&\rightarrow&\ps{\Cl^{\leftrightarrow}}_V\\
(\mu, \nu)&\mapsto&(\mu^*, \nu^*)
\ea
where $\mu^*(V):=(\mu(V))^*$ However, if $\mu$ is order preserving it is not necessarily the case that $\mu^*$ is. 
This is related to the same problem which prevents us from defining subtraction in $\ps{\Rl^{\leftrightarrow}}$.
\item {\it Sum} 
\begin{Definition}
The sum operation is defined to be a map $+:\ps{\Cl^{\leftrightarrow}}\times \ps{\Cl^{\leftrightarrow}}\rightarrow \ps{\Cl^{\leftrightarrow}}$ such that for each $V\in \mv(\mh)$ we have
\ba
+_V:\ps{\Cl^{\leftrightarrow}}_V\times \ps{\Cl^{\leftrightarrow}}_V&\rightarrow& \ps{\Cl^{\leftrightarrow}}_V\\
\big((\mu, \nu), (\mu^{'}, \nu^{'})\big)&\mapsto&(\mu+\mu^{'}, \nu+\nu^{'})
\ea
Where $(\mu+\mu^{'}, \nu+\nu^{'})(V^{'}):=(\mu+\mu^{'}(V^{'}),  \nu+\nu^{'}(V^{'}))=(\mu(V^{'})+\mu^{'}(V^{'}),  \nu(V^{'})+\nu^{'}(V^{'}))$ for all $V\in \mv(\mh)$.
\end{Definition}
It is straightforward to see that, as defined above, the maps $(\mu+\mu^{'},  \nu+\nu^{'})$ are a pair of order preserving and order reversing maps.
\item {\it Multiplication}. Similar to the case for $\ps{\Rl^{\leftrightarrow}}$ we can not define multiplication.
\item {\it Subtraction}. Similar to the case for $\ps{\Rl^{\leftrightarrow}}$ we can not define subtraction.
\end{enumerate}
Because of the above properties $\ps{\Cl^{\leftrightarrow}}$ is only a monoid (semigroup with a unit). However, as we will see later on, it is possible to transform such a semigroup into a group through the process of Grothendieck k-extension, obtaining the object $k(\ps{\Cl^{\leftrightarrow}})$. As we will see, this object can be seen as a vector space over $\ps{\Rl}$. 

It would be of particular interest to understand what and if $\ps{\Cl^{\leftrightarrow}}$ has any topological properties. To this end we would have to define $
\Cl^{\leftrightarrow}$ as an internal locale.
Work in this direction has been partially done in \cite{dutch}.  Here the authors introduce an alternative {\it internal} formulation of quantum topos theory. In particular, given a $C^*$-algebra $\a$, they define the internal\footnote{Note that in this internal aproach one is working with co-presheaves instead of presheaves.} $C^*$ algebra $\tilde{\a}\in Sets^{\mathcal{C}(\a)}$ where $\mathcal{C}(\a)$ is the category of abelian sub-algebras of $\a$ ordered by inclusion. 
Given such an internal algebra they construct its spectrum $\ts$ and show that it is an internal locale. This enables them to define the (internal) topos of sheaves $Sh(\ts)$. 
They then construct the locale $\ps{\Rl}_{Sh(\ts)}$ which has as associated sheaf the sheaf $pt(\ps{\Rl}_{Sh(\ts)})$ of Dedekind Reals in $Sh(\ts)$. The detailed way in which $pt(\ps{\Rl}_{Sh(\ts)})$ is defined can be found in \cite{maklane}. 
Given the internal local $\ps{\Rl}_{Sh(\ts)}$ it is now possible to construct the internal locale $\ps{\Cl}_{Sh(\ts)}$ whose associated sheaf would be the complex number object in $Sh(\ts)$ defined by $pt(\ps{\Cl}_{Sh(\ts)})\simeq pt(\ps{\Rl}_{Sh(\ts)})\times pt(\ps{\Rl}_{Sh(\ts)}) $. A very in depth analysis of the complex number object can be found in \cite{Banasche1}, \cite{Banasche2} and \cite{chrisa}.
What is interesting for us is that the object $pt(\ps{\Rl}_{Sh(\ts)})$ was shown (\cite{dutch}) to be related to $\ps{\Rl^{\leftrightarrow}}$ in the case of quantum theory. This would suggest that the object $\ps{C}^{\leftrightarrow}$ is related to $\ps{\Cl}_{Sh(\ts)}$. If this were the case, it would help unveil what, if any, topological properties $\ps{C}^{\leftrightarrow}$ has. Such an analysis, however, is left for future publications.

\subsection{Domain-Theoretic Structure}
In the recent paper \cite{domain} the authors show how the quantity value object $\ps{\Rl^{\leftrightarrow}}$ can be given a domain theoretic structure. 
This then results in the fibres $\ps{\Rl^{\leftrightarrow}}_V$ being almost-bounded directed-complete posets (see later on for the appropriate definitions).
We will now give a brief description on how the results given in \cite{domain} can be generalised for the complex quantity value object $\ps{\Cl^{\leftrightarrow}}$ defined above.
As a first step we give the definition of a closed rectangle in the complex plane $\Cl$
\begin{Definition}
The set 
\be
T_{\alpha,\beta}=\{z\in\Cl|a\leq Re(z)\leq c, b\leq Im(z)\leq d\, ; a,b,c,d\in \Rl\}
\ee
defines a closed rectangle in $\Cl$. Denoting $\alpha=a+ib$ and $\beta=c+id$, then the above closed rectangle is defined by the two points $(\alpha, \beta)$.
\end{Definition}
It is clear from the definition that $\alpha\leq \beta$ for the ordering in $\Cl$ defined in \ref{def:order}. However this definition of closed rectangles does not account for the general case in which $\alpha\leq \beta$ if $a+b\leq c+d$ but not necessarily $a\leq c$ and $b\leq d$. To remedy this we slightly change the above definition of a closed rectangle, obtaining
\begin{Definition}
Given any two points $\alpha, \beta\in \Cl$ such that $\alpha\leq \beta$ according to Definition \ref{def:order}, then the general closed rectangle ``spanned'' by them is 
\be
[\alpha, \beta]:=\{z\in\Cl|Re(z)\in [Re(\alpha), Re(\beta)] \wedge Im(z)\in[Im(\alpha),Im(\beta)]\}
\ee
where $ [Re(\alpha), Re(\beta)] $ is the closed line interval spanned by $Re(\alpha)$ and $Re(\beta)$ and similarly for $[Im(\alpha),Im(\beta)]$.
Clearly, for a given $\alpha,\beta$, $T_{\alpha,\beta} \equiv [\alpha,\beta]$.
\end{Definition}
Following the discussion of Section \ref{sec:complex}, it is clear that in the topos approach normal operators (as well as self-adjoint operators) are assigned an interval of values which are called ``unsharp values''.  The set of such ``unsharp complex values'' is defined as follows
\be
\mathbf{I}\Cl=\{[\alpha,\beta]|\alpha, \beta\in \Cl\,; \alpha\leq\beta\}
\ee
Clearly $\Cl\subset \ic$ and it consists of all those intervals for which $\alpha=\beta$.\\
The claim is that, similarly as done in \cite{domain} for $\mathbf{I}\Rl$, the set $\ic$ is a {\it domain} whose definition is given as follows:
\begin{Definition}
A {\it domain} $\langle D, \sqsubseteq \rangle$ is a poset such that i) any directed set\footnote{ A set $P$ is directed if for any $x, y\in P$ there exists a $z\in P$ such that $x,y\sqsubseteq z$} has a supremum, i.e. it is a directed-complete poset (dcpo); ii) it is continuous: for any $d\in D$ one has $\bigsqcup^{\uparrow}\SSEarrow y=y$ where $\SSEarrow y:=\{x\in D|x\ll y\}$\footnote{Here $\bigsqcup^{\uparrow}$ indicates the supremum of a directed set and the relation $x\ll y$ indicates that $x$ {\it approximates} $y$. In particular $x\ll y$ if, for any directed set $S$ with a supremum, then $y\sqsubseteq  \bigsqcup^{\uparrow} S\Rightarrow \exists s\in S:x\sqsubseteq s$.}. 
\end{Definition}
\noindent For the case of $\ic$ we obtain the definition
\begin{Definition}
The {\it complex interval domain} is the poset of closed rectangles in $\Cl$ partially ordered by reverse inclusion
\be
\ic:=\langle\{[\alpha, \beta] | \alpha, \beta\in \Cl\}, \sqsubseteq:=\supseteq \rangle
\ee
where $[\alpha,\beta]\sqsubseteq [\alpha_1,\beta_2]$ if $[\alpha_1,\beta_2]\subseteq [\alpha,\beta]$ in the sense that any complex number $z$ lying in the rectangle $[\alpha_1,\beta_1]$ will also lie in the rectangle $[\alpha,\beta]$.
\end{Definition}
Following the discussion given in \cite{domain}, we can denote a rectangle as $x=[x_{-}, x_{+}]$ where $x_-$ represents the left ``end point'' while $x_+$ represents the right ``end point''. In this way, for each function $f:X\rightarrow \ic$ there corresponds a pair of functions $f_-, f_+:X\rightarrow \Cl$ defined as $f_{\pm}(x):=(f(x))_{\pm}$ such that  $f_-\leq f_+$ (pointwise order) and  $f(x)=[f_-(x), f_+(x)]$. 
Conversely, for each pair of functions $g\leq h:X\rightarrow\Cl$ there corresponds a function $f:X\rightarrow\ic$ such that $f_-=g$ and $f_+=h$. The decomposition of each map $f:X\rightarrow \ic$ into two maps can be explicitly stated by the following diagram
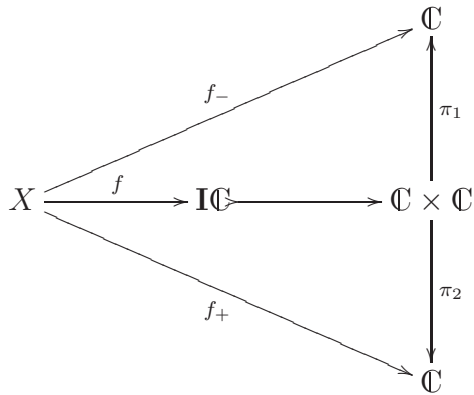
\[\xymatrix{
&&&&\Cl\\
&&&&\\
X\ar[rr]^{f}\ar[rrrruu]^{f_-}\ar[rrrrdd]_{f_+}&&\ic\ar@{>->}[rr]&&\Cl\times\Cl\ar[uu]_{\pi_1}\ar[dd]^{\pi_2}\\
&&&&\\
&&&&\Cl\\
}\]
In order for $\ic$ to be a well defined domain we need to show that it satisfies the definition of a domain. To achieve this we will need the generalised nested rectangle theorem
\begin{Theorem}
Given a sequence $\{[\alpha, \beta]_n\}$ of nested generalised closed rectangles (as defined above ) such that $\lim_{n\rightarrow\infty}l([\alpha, \beta]_n)=0$\footnote{Here 
$l([\alpha, \beta]_n)$ represents the length of the largest side of the rectangle and thus is defined as $l([\alpha, \beta]_n)=Max\{|a_n-c_n|, |b_n-d_n|\}$ where $\alpha=a+ib$ and $\beta=c+id$.} then the following conditions hold:
\begin{enumerate}
\item 
$[\alpha, \beta]=\bigcap_{i\in N} [\alpha, \beta]_i=z_0$ for some $z_0\in \Cl$
\item Given $\epsilon > 0$ there is an $m\in\Nl$ such that 
\be
[\alpha, \beta]_n\subset \{z:|z-z_0|<\epsilon\}\forall n> m
\ee
\end{enumerate}

\end{Theorem} 
\begin{proof}
To prove the first condition we resort to the theorem of nested intervals in $\Rl$. In particular the rectangle $[\alpha, \beta]_n$ has as boundary lines $Rel(z)=a_n$, $Rel(z)=c_n$, $Im(z)=b_n$, $Im(z)=d_n$ where $\alpha_n=a_n+ib_n $ and $\beta_n=c_n+id_n$. 
Therefore we have two sequences of nested intervals $R_n=[a_n, c_n]$ and $I_n=[b_n, d_n]$ such that $[\alpha, \beta]_n=R_n\times I_n$. From the theorem of nested sequences of intervals it follows that $\bigcap_{i\in N} R_n=a$ and $\bigcap_{i\in N} I_n=d$ thus $\bigcap_{i\in N} [\alpha, \beta]_n=a+id=z_0$.\\
To prove condition 2) we choose an arbitrary element $z\in [\alpha, \beta]_n$. Then, given the existence of $ [\alpha, \beta]=\bigcap_{i\in N}  [\alpha, \beta]_i=z_0$ it follows that
\be
|z-z_0|^2\leq |\alpha_n-\beta_n|^2=|a_n-c_n|^2+|b_n-d_n|^2=2(l( [\alpha, \beta]_n))^2
\ee
Since, by assumption, $\lim_{n\rightarrow\infty}l([\alpha, \beta]_n)=0$, it follows that there exists an $m$ such that given $\epsilon> 0$
\be
\sqrt{2}l( [\alpha, \beta]_n)<\epsilon\;,\forall n> m
\ee
Therefore
\be
|z-z_0|\leq\sqrt{2}l( [\alpha, \beta]_n)<\epsilon
\ee
The above holds for any $z\in  [\alpha, \beta]_n$  therefore $ [\alpha, \beta]_n\subset \{z:|z-z_0|<\epsilon\}$.
\end{proof}
The reason we went through the trouble of stating the nested rectangle theorem is because we will use it when defining the supremum of directed subsets $S\in \ic$. In particular, given such a directed set we then have
\be
\bigsqcup\hspace{.0001pt}^{\uparrow} [\alpha, \beta]=\bigcap  [\alpha, \beta]=[sup\{x_-|x\in  [\alpha, \beta]\}, inf\{x_+|x\in  [\alpha, \beta]\}]
\ee
Thus for a sequence $ [\alpha, \beta]_n$ of nested rectangles we simply get a point $z_0$ as the supremum. \\
Given the above, we can define the relation $\ll$ in $\ic$ as follows
\begin{Definition}
Given any two rectangles $x, y$ then \be
x\ll y\text{ iff } (x_-< y_-)\wedge(y_+< x_+)
\ee
\end{Definition}
\noindent For $\ic $ to be a well defined domain we need to show that $\bigcup^{\uparrow}\SSEarrow y=y$. Since $\SSEarrow y:=\{x\in D|x\ll y\}$ then clearly
\be
\bigcup\hspace{.0001pt}^{\uparrow}\SSEarrow y=\bigcap \{x\in D|x\ll y\}=\bigcap\{x\in\ic | (x_-< y_-)\wedge(y_+< x_+)\}=y
\ee
Since $\ic$ is a domain it is a continuous poset and as such it comes with a Scott topology\footnote{We recall that, given a poset $\langle P, \leq\rangle$ a subset $G$ is said to be Scott-open if i) $x\in G\wedge x\leq y\Rightarrow y\in G$; ii) for any directed set $S$ with supremum then $\bigsqcup^{\uparrow}S\in G\Rightarrow \exists s\in S|s\in G$. In other words all supremums in $G$ have a non-empty intersection with $G$. } whose basis opens are 
\be
\Ssearrow[\alpha,\beta]:=\{[\gamma,\delta]|\alpha<\gamma\leq\beta<\delta\}=\{\sigma\in\ic|\sigma\subseteq (\alpha, \beta)\}
\ee
where $(\alpha, \beta):=\{z\in \Cl|\alpha<z<\beta\}$ represents the general open rectangle ``spanned'' by $\alpha, \beta$\footnote{
Given this topology then it is clear that, as topological spaces $\Cl\simeq max\ic$ where $\Cl$ is equipped with the (general) open rectangles topology and $max\ic$ has the topology inherited by $\ic$. The homeomorphisms can be see by the fact that $
\Ssearrow[\alpha,\beta]\cap max\ic=\{\sigma\in\ic|\sigma\subseteq (\alpha, \beta)\}\cap max\ic=\{\{\beta\}|\beta\in (\alpha, \beta)\}=(\alpha, \beta)
$}.
Recalling that any map $f:X\rightarrow\ic$ can be decomposed into a left and right part $f_- $ and $f_+$, respectively, such that $f_-\leq f_+$, we note that $f$ is order preserving iff $f_-$ is order preserving and $f_+$ is order reversing. This, similar to the case for the real quantity value object, suggests we re-write the complex valued object as follows
\begin{Definition}
The complex value object $\ps{\Cl^{\leftrightarrow}}$ is defined on
\begin{enumerate}
\item Objects: for each $V\in\mv(\mh)$ we obtain the set
\be
\ps{\Cl^{\leftrightarrow}}_V=\{f:\downarrow V\rightarrow \ic|f \text{order-preserving}\}
\ee
\item Morphisms: given $i_{V'V}:V'\subseteq V$ the corresponding morphism is
\ba
\ps{\Cl^{\leftrightarrow}}(i_{V'V}):\ps{\Cl^{\leftrightarrow}}_V&\rightarrow& \ps{\Cl^{\leftrightarrow}}_{V'}\\
f&\mapsto&f|_{\downarrow V'}
\ea
\end{enumerate}
\end{Definition}

We now utilise the analogue of proposition 4.2 in \cite{domain} as applied for the complex number object. The proof is identical to the case of the real valued object so we will omit it.
\begin{Proposition}
The global elements of $\ps{\Cl^{\leftrightarrow}}$ are in bijective correspondence with order-preserving functions from $\mv(\mh)$ to $\ic$.
\end{Proposition}
\noindent The new reformulation of the complex valued object together with the above proposition imply that that, for each $V\in\mv(\mh)$
\ba
\ps{\Cl^{\leftrightarrow}}_V&=&OP(\downarrow V, \ic)\\
\Gamma\ps{\Cl^{\leftrightarrow}}&=&OP(\mv(\mh), \ic)
\ea
where $OP(\downarrow V, \ic)$ indicates the set of order preserving functions from the poset $\downarrow V$ to $\ic$.\\
By equipping both $\downarrow V$ and $\mv(\mh)$ with the Alexandroff topology and utilising propositions 4.3 and 4.4 in \cite{domain} we arrive at the following results: for each $V\in\mv(\mh)$
\ba
\ps{\Cl^{\leftrightarrow}}_V&=&C(\downarrow V, \ic)\text{ and } \ps{\Cl^{\leftrightarrow}}_V \text{ is an almost complete dcpo}\\
\Gamma\ps{\Cl^{\leftrightarrow}}&=&C(\mv(\mh), \ic)\text{ and } \Gamma\ps{\Cl^{\leftrightarrow}}_V \text{ is an almost complete dcpo}
\ea
Thus, with the use of domain theory, we can understand the precise structure of the complex valued quantity object.

\section{Normal Operators in Terms of Functions of Filters}
We would now like to check whether the discussion done in \cite{Doering2008} regarding the relation between self-adjoint operators and functions on filters still holds for the case of normal operators. This should indeed be the case since there exists a spectral theorem for normal operators, and this is all that is really needed.

First of all we recall that, given a lattice $L$ it is possible to define a map from the Stone spectrum $\mathcal{Q}(L)$ (see Section \ref{sec:lattice} in the Appendix for the relevant definitions) to the reals $\Rl$ as follows:
\begin{Definition}
Given a bounded spectral family $E:\Rl\rightarrow L$ in a complete lattice $L$, then it is possible to define a function
\ba
f_E:\mathcal{Q}(L)&\rightarrow&\Rl\\
\mb&\mapsto& inf\{\lambda|E_{\lambda}\in\mb\}
\ea
Such a function is called an observable function corresponding to $E$.
\end{Definition}
Of particular relevance to us is when the lattice in question is the complemented distributive lattice of projection operators in a von Neumann algebra. In that case, for each self-adjoint operators, we obtain a corresponding observable function. 

We now would like to extend, in a meaningful way, the above definition to normal operators. A first guess would be the following definition:
\begin{Definition}
Given a normal operator $\hat{A}$ with spectral family $\{\hat{E}^{\hat{A}}_{\lambda}\}_{\lambda\in \Cl}$ the corresponding observable function is
\ba
f_{\hat{A}}:\mathcal{Q}(P(V))&\rightarrow&\Cl\\
\mb&\mapsto& inf\{\lambda|\hat{E}^{\hat{A}}_{\lambda}\in\mb\}
\ea
\end{Definition}

Where the {\it infimum} is defined according to the ordering in definition \ref{def:order}.
From now on we will call observable functions which correspond to normal operators ``{\it normal-observable functions}'', to distinguish them from observable functions as related to self-adjoint operators.

We now have to reproduce the theorems done in \cite{Groote2007a} which show that the above definition is well defined. In particular we need to prove the following:
\begin{Theorem}
Given a normal operator $\hat{A}$ in a von Neumann algebra $N$ with associated normal-observable function $f_{\hat{A}}:\mathcal{Q}(P(V))\rightarrow\Cl$, then 
\be
im f_{\hat{A}}=sp(\hat{A})
\ee
\end{Theorem}
\begin{proof}
 By contradiction, we start by assuming that there
exists some $\lambda_0 \notin sp(\hat{A})$ which also obeys
$\lambda_0 \in im(f_{\hat{A}})$. \\
We know that $\Cl$ is a metric space through the following metric
\ba
d:\Cl\times\Cl\rightarrow\Rl\\
(z_1, z_2)\mapsto |z_1-z_2|
\ea
We can then define an open ball around any point $z_0\in \Cl$ as
\be
D_{\varepsilon}(z_0)=\{z:|z-z_0|<\varepsilon\}
\ee
for some $\varepsilon>0$. 

Now, we know that the spectrum of a normal operator $\hat{A}$ consists of all $\lambda\in \Cl$ such that the spectral family $\hat{E}^{\hat{A}}_{\lambda_i}$ is non-constant on every neighbourhood of $\lambda$. Thus, since we have assumed that $\lambda_0 \in im(f_{\hat{A}})$ but $\lambda_0 \notin sp(\hat{A})$, it is reasonable to assume that there exists an open ball $D_{\varepsilon}(z_0)$ which is a neighbourhood of $\lambda$ and where 

\begin{equation}
\label{ref:alllambda}
\forall \lambda \in D_{\varepsilon}(z_0) : E_{\lambda}^A = E_{\lambda_0}^A
\end{equation}
i.e. the family of operators is constant on such a neighbourhood. Now if $\mb\subseteq f_{\hat{A}}^{-1}(\lambda_0)$, then from the definition of $f_{\hat{A}}$ and the fact that $\forall \lambda \in D_{\varepsilon}(z_0) : E_{\lambda}^A = E_{\lambda_0}^A$, then $f_{\hat{A}}A(\mb)\leq inf (D_{\varepsilon}(z_0) )$. In fact we have that if
\be
\mb\subseteq f^{-1}_{\hat{A}}(\lambda_0)
\ee
then 
\be
f_{\hat{A}}(\mb)\leq f_{\hat{A}} f^{-1}_{\hat{A}}(\lambda_0)=\lambda_0
\ee
But given equation \ref{ref:alllambda} then 
\be
f_{\hat{A}}(\mb)\leq inf (D_{\varepsilon})
\ee
However since we had assumed that $\lambda_0\notin sp(A)$ then $inf (D_{\varepsilon})\notin sp(A)$ we obtain a contradiction and 
$im f_A\subseteq sp(A)$.

We now want to show that $sp(\hat{A})\subseteq im f_{\hat{A}}$.To this end we need the notion of a limit of a sequence of complex numbers.
\begin{Definition}
Given a sequence of complex numbers $z_n$ then the limit is $z$ iff, for all $\varepsilon >0$ there exists a natural number $m$, such that $n > m$ implies $|z-z_n|<\varepsilon$.
\end{Definition}
Another way of defining a limit of a series of complex numbers is as follows:
\begin{Definition}
If we denote the real and imaginary part of a complex number as $Re(z_k) = x_k$ and $Im(z_k) = y_k$, then the sequence of complex numbers $z_1, z_2,\cdots$ has a limit, if and only if, the sequences of real numbers $x_1, x_2,\cdots $ and $y_1, y_2,\cdots$ have limits. Then we obtain
\be
lim (z_k) = lim (x_k) + i \cdot lim (y_k).
\ee
\end{Definition}
Another useful definition is:
\begin{Definition}
A sequence of complex numbers $z_n$ converges to some limit iff for all $\varepsilon >0$ there exists some natural number $m$, such that for $n, p>m$   $|z_m-z_n|<\varepsilon$.
\end{Definition}
Given these notions we now assume that $\lambda_0 \in sp(\hat{A})$. We then have two cases
\begin{enumerate}
\item [i.] We have a decreasing sequence $(\lambda_n)_{n\in\Nl}$ such that the limit of this sequence as defined above is $\lambda_0$ and 
for all $n$, $E_{\lambda_{n+1}}^A < E_{\lambda_n}^A$.
In this case, select a $\mb$ such that $ E^A_{\mu}-E_{\lambda_0}^A \in \mb$ for all $\mu>\lambda_0$. This implies that for all $\lambda >\lambda_o$ then $ E^A_{\lambda}-E_{\lambda_0}^A \in \mb$ and since $E^A_{\lambda}> E^A_{\lambda}-E_{\lambda_0}^A$ and $\mb$ is upwards closed then $ E^A_{\lambda}\in \mb$. On the other hand $E_{\lambda_0}^A \notin \mb$ which implies that $f_{\hat{A}}(\mb) = \lambda_0$.
Therefore, for any $\lambda_0 \in sp(\hat{A})$, $\lambda_0 \in f_{\hat{A}}$.
\item[ii.]
The only other case is where $E_{\lambda}^A < E_{\lambda_0}^A$ for $\lambda < \lambda_0$.
In this case we can just take a quasipoint which contains $ E_{\lambda_0}^A$ but does not contain any $E_{\lambda}^A$
for  $\lambda < \lambda_0$.
Then, by the definition of infimum, $f_{\hat{A}}(\mb) = \lambda_0$, and again,  for any $\lambda_0 \in sp(\hat{A})$, $\lambda_0 \in f_{\hat{A}}$.

\end{enumerate}

\end{proof}
\begin{Theorem}\label{the:contin}
Given a normal operator $\hat{A}$, then the observable function $f_{\hat{A}}:\mathcal{Q}(P(V))\rightarrow \Cl$ is continuous.
\end{Theorem}
This proof will be a generalisation of the proof of theorem 2.4 in \cite{Groote2007c} as applied to complex numbers. To carry out this proof we need to recall a few facts about continuity. 
\begin{Theorem}
\textbf{Uniform convergence theorem}. Let $S$ be a set of complex numbers and $\{f_n\}$ a sequence of continuous functions on $S$. If $\{f_n\}$ converges uniformly, then the limit function $f$ is also continuous on $S$.
\end{Theorem}
A proof of this theorem can be found in \cite{Lang1999}. We will use this result to prove theorem \ref{the:contin}.\\The other result we need is the definition of uniform convergence for complex valued functions. 
\begin{Definition}
Given a set $S$ of complex numbers, then the sequence $\{f_n\}$ of functions on $S$ converges uniformly on $S$ if there exists a function $f$ on $S$ such that, given $\varepsilon>0$ there exists and $N$ such that for $n\geq N$
\be
||f_n-f||<\varepsilon
\ee
where $||f||=sup_{z\in S}|f(z)|$.
\end{Definition}
We are now ready to prove the theorem \ref{the:contin}.
\begin{proof}
First of all we know that $im f_A=sp(A)$. Define $\gamma:=min(sp(A))$ and $\beta:=max(sp(A))$ where again the ordering was defined above. Given a real number $\varepsilon >0$ we then define two open intervals as follows:\\
First we construct the circular annulus 
\be
S_a:=\{z:|a|-\varepsilon<|z|\leq|a|\}
\ee
This is neither open nor closed and it does not contain $a$. To make it an open set we consider the interior whose construction utilises the following definition:
\begin{Definition}
A point $z_0$ is an interior point of a set $S$ if $\exists$ $\varepsilon>0$ such that $D_{\varepsilon}(z_0)\subseteq S$.
\end{Definition}
We denote by $int(S)$ the interior of a set $S$. We then choose $\lambda_0\in int(S_a)$. Similarly we define
\be
S_b:=\{z:|b|<|z|\leq|b|-\varepsilon\}
\ee
and choose $\lambda_n\in int(S_b)$.\\
We then construct 
\be
S_{a,b}:=\{z:|b|<|z|\leq|a|\}
\ee
and consider $\lambda_1, \cdots\lambda_n\in\overline{S_{a,b}}=S_{a,b}\cup\delta S_{a, b}$ (here $\delta S_{a, b}=\{z:|z|=|b|\}\cup\{z:|z|=|a|\}$), such that $\lambda_{k-1}<\lambda_k$ and $|\lambda_k-\lambda_{k-1}|<\varepsilon$ for $k=1,\cdots ,n$. Next we define 
\be
int(S_{*}):=int\big(\{z:|\lambda_{k-1}|<|z|<|\lambda_k|\}\big)
\ee
The let $\lambda^*\in int(S_{*})$ for $k=1,\cdots ,n$. We then define an operator which is $\varepsilon$ dependent:
\be
A_{\varepsilon}:=\sum_{k=1}^n\lambda^*_k(\hat{E}^{A}_{\lambda_k}-\hat{E}^A_{\lambda_{k-1}}):=\sum_{k=1}^n\lambda^*_kP_k
\ee
where $P_k:=\hat{E}^{A}_{\lambda_k}-\hat{E}^A_{\lambda_{k-1}}$. 
These $\varepsilon$-dependent operators give rise to a sequence of continuous functions as shown in proposition 2.6 in \cite{Groote2007}. These are
\be
f_{A_{\varepsilon}}=\sum_{k=1}^n\lambda^*_k\chi_{\mathcal{Q}_{\hat{E}^A_{\lambda_k}(P(V))}/\mathcal{Q}_{\hat{E}^A_{\lambda_{k-1}}(P(V))}}
\ee
Where $\mathcal{Q}_{\hat{E}^A_{\lambda_k}}(P(V))=\{\mb|\hat{E}^A_{\lambda_k}\in \mb\}$ and $\chi$ is the usual characteristic function.
Then for any $\mb\in \mathcal{Q}(P(V))$, $\mb\in \mathcal{Q}_{\hat{E}^A_{\lambda_k}(P(V))}/\mathcal{Q}_{\hat{E}^A_{\lambda_{k-1}}(P(V))}$ for only one $k$, therefore
\be
f_{A_{\varepsilon}}(\mb)=\lambda^*_k
\ee
If we now go back to our original operator $A$ we notice that
\be
f_A(\mb)\in \overline{S_*}
\ee
thus
\be
|f_A(\mb)-f_{A_{\varepsilon}}(\mb)|<\varepsilon
\ee
Since the $\mb$ was arbitrary we obtain the desired result
\be
||f_A-f_{A_{\varepsilon}}||<\varepsilon
\ee
i.e. $f_A$ is continuous.
\end{proof}
\begin{Definition}
Given a von Neumann algebra $\mathcal{N}$, then the set of all observable functions $\mathcal{Q}(P(V))\rightarrow \Cl$ is denoted by $\mathcal{O}(\mathcal{N})$.

\end{Definition}

\begin{Theorem}\label{the:restriction}
Given an abelian von Neumann algebra $\mathcal{N}$, then the mapping $\mathcal{N}\rightarrow C(\mathcal{Q}(P(V)), \Cl)$;  $A\mapsto f_A$ is, up to the isomorphisms $C(\mathcal{Q}(P(\mathcal{N})), \Cl)\rightarrow C(\Omega(\mathcal{N}),\Cl)$, where $\Omega(\mathcal{N})$ is the Gel'fand spectrum of $\mathcal{N}$, the restriction of the Gel'fand transform to $\mathcal{N}_{normal}$.
\end{Theorem}
This theorem was proved in \cite{Groote2007} (Theorem 2.9) for self-adjoint operators. The generalisation to normal operators is very straightforward and rests on proposition 2.17 in \cite{Groote2007} and theorem \ref{the:contin} above.

The above theorem has an important consequence. In particular, given a normal operator $\hat{A}=\hat{C}+i\hat{B}$ then the Gel'fand transform is
\be
F_{\hat{A}}=F_{\hat{C}}+iF_{\hat{B}}
\ee
However, since $F_{\hat{A}}=f_{\hat{A}}$ it follows that
\be\label{equ:aib}
f_{\hat{A}}=f_{\hat{C}}+if_{\hat{B}}
\ee

A straightforward consequence of theorem \ref{the:restriction} is
\begin{Theorem}
Given a von Neumann algebra $\mathcal{N}$ with $\mathcal{O}(\mathcal{N})$ the set of all observable functions, then 
\be
\mathcal{O}(\mathcal{N})=C(\mathcal{Q}(P(V)), \Cl)
\ee
iff $\mathcal{N}$ is abelian. Where here $C(\mathcal{Q}(P(V)), \Cl)$ denotes the set of all bounded continuous functions $\mathcal{Q}(P(V))\rightarrow  \Cl$.
\end{Theorem}
\subsection{Relation Between Observable Functions and Normal Operator Functions}
In the previous section we have already seen that it is possible to deduce that, given a normal operator $\hat{A}=\hat{C}+i\hat{B}$, the corresponding normal-observable function $f_{\hat{A}}$ is such that $f_{\hat{A}}=f_{\hat{C}}+if_{\hat{B}}$, where $f_{\hat{C}}$ and $f_{\hat{B}}$ are the observable functions of $\hat{C}$ and $\hat{B}$. In this section, however, we would like to give a more constructive proof of this result. To this end we first of all have to introduce the diagonal map between two sets as follows\footnote{We can obviously turn the set $\mathcal{Q}(\mathcal{P}(V))$ into a category by considering only the identity morphisms, then the map $\bigtriangleup$ would be the diagonal functor. We then consider $\Cl$ as a poset ordered by the ordering defined in \ref{def:order} and the product $\Rl\times\Rl$ as the product poset. In this line of reasoning then the map $+_{\Cl} \circ \langle f_{\hat{C}},f_{\hat{B}}\rangle\circ \bigtriangleup$ becomes trivially a functor.  }
\ba
\bigtriangleup: \mathcal{Q}(\mathcal{P}(V))&\rightarrow&\mathcal{Q}(\mathcal{P}(V))\times \mathcal{Q}(\mathcal{P}(V))\\
\mb&\mapsto&(\mb,\mb)
\ea
As can be seen from the definition, the diagonal map simply assigns to each object $\mb$ a pair of copies of itself $(\mb,\mb)$. Given such a map we now define the following:
\ba
\mathcal{Q}(\mathcal{P}(V))&\xrightarrow{\bigtriangleup}&\mathcal{Q}(\mathcal{P}(V))\times \mathcal{Q}(\mathcal{P}(V))\xrightarrow{\langle f_A, f_B\rangle}\Rl\times\Rl\xrightarrow{+_{\Cl}}\Cl\\
\mb&\mapsto&(\mb,\mb)\mapsto \left(inf\{\varepsilon|\hat{E}^{\hat{C}}_{\varepsilon}\in \mathcal{B}\}, inf\{\eta|\hat{E}^{\hat{B}}_{\eta}\in \mathcal{B}\} \right) \mapsto  \left( inf\{\varepsilon|\hat{E}^{\hat{C}}_{\varepsilon}\in \mathcal{B}\}+(i) \cdot inf\{\eta|\hat{E}^{\hat{B}}_{\eta}\in \mathcal{B}\}\right)\nonumber
\ea
where the map $+_{\Cl}$ is the isomorphism defined by
\ba
+_{\Cl}\Rl\times\Rl&\rightarrow&\Cl\\
(a,b)&\mapsto&a+ib
\ea
Given the above functor we attempt the following conjecture:
\begin{Conjecture}
Given the observable functions
\ba
f_C:\mathcal{Q}(\mathcal{P}(V))&\rightarrow& \Rl\\
\mathcal{B}&\mapsto&inf\{\varepsilon|\hat{E}^{\hat{C}}_{\varepsilon}\in \mathcal{B}\}
\ea
and 
\ba
f_B:\mathcal{Q}(\mathcal{P}(V))&\rightarrow& \Rl\\
\mathcal{B}&\mapsto&inf\{\eta|\hat{E}^{\hat{B}}_{\eta}\in \mathcal{B}\}
\ea
then the following diagram commutes
\[\xymatrix{
\mathcal{Q}(\mathcal{P}(V))\ar[rr]^{\Delta}\ar[dd]^{f_A}&&\mathcal{Q}(\mathcal{P}(V))\times\mathcal{Q}(\mathcal{P}(V))\ar[dd]^{\langle f_C, f_B\rangle}\\
&&\\
\Cl&&\Rl\times \Rl\ar[ll]^{+_{\Cl}}\\
}\]

\end{Conjecture}

Thus we are interested in showing, for each normal operator $\hat{A}$, that the associated normal-observable function $f_{\hat{A}}$ is equivalent to the composite map $+_{\Cl} \circ \langle f_{\hat{C}},f_{\hat{B}}\rangle\circ \bigtriangleup$  iff $\hat{A}=\hat{C}+i\hat{B}$.

To this end we recall that, a given $\mathcal{B}\in \mathcal{Q}(\mathcal{P}(V))$ we have that for $\hat{A}=\hat{C}+i\hat{B}$
\ba
f_A(\mathcal{B})&:=&inf\{\gamma|\hat{E}^{\hat{A}}_{\gamma}\in \mathcal{B}\}\\
&=&inf\{\varepsilon+i\eta|\hat{E}^{\hat{C}}_{\varepsilon}\hat{E}^{\hat{B}}_{\eta}\in \mb\}\\
&=&inf\{\varepsilon+i\eta|\hat{E}^{\hat{C}}_{\varepsilon}\wedge \hat{E}^{\hat{B}}_{\eta}\in \mb\}\\
&=&inf\{\varepsilon+i\eta|\hat{E}^{\hat{C}}_{\varepsilon}\in \mb\text{ and }\hat{E}^{\hat{B}}_{\eta}\in \mb\}
\ea
The third equation follows from the fact that a maximal dual ideal is closed under taking the $\wedge$ operation. 

However, utilising the ordering of $\Cl$ defined in \ref{def:order}\footnote{$(\lambda_1=\varepsilon_1+i\eta_1)\leq (\lambda_2=\varepsilon_2+i\eta_2)$ if $\varepsilon_1+\eta_1\leq \varepsilon_2+\eta_2$.} we obtain the following theorem: 
\begin{Theorem}
Given a normal operator $\hat{A}=\hat{C}+i\hat{B}$ and the associated self-adjoint operator $\hat{C}+\hat{B}$ then
\be
f_{\hat{A}}(\mb)=inf\{\gamma|\hat{E}^{\hat{A}}_{\gamma}\in \mathcal{B}\}\text{ iff } f_{\hat{C}+\hat{B}}(\mb)=inf\{\varepsilon+\eta|\hat{E}^{\hat{C}}_{\varepsilon}\in \mb\text{ and }\hat{E}^{\hat{B}}_{\eta}\in \mb\}
\ee

\end{Theorem}

\begin{proof}
Let us denote $inf\{\gamma|\hat{E}^{\hat{A}}_{\gamma}\in \mathcal{B}\}=\inf\{B\}$ and $inf\{\varepsilon+\eta|\hat{E}^{\hat{C}}_{\varepsilon}\in \mb\text{ and }\hat{E}^{\hat{B}}_{\eta}\in \mb\}=inf\{C\}$. We know that $\inf\{B\}=a+ib$ is such that for all $\sigma+i\gamma\in B$, $\inf\{B\}\leq \sigma+i\gamma$. \\
However, by the definition of ordering we know that  $\inf\{B\}=a+ib\leq \sigma+i\gamma$ if $a+b\leq \sigma+\gamma$. Since $\sigma+i\gamma\in B$ iff $ \sigma+\gamma\in C$, it follows that $\inf\{B\}=a+ib\leq \sigma+i\gamma$ if $inf\{C\}=a+b\leq \sigma+\gamma$. \\
On the other hand if $a+b=inf\{C\}$ then by definition $a+b\leq \sigma+\gamma$ for all $\sigma+\gamma\in C$. Again from the definition of ordering it then follows that $a+ib=inf\{B\}$.
Therefore
\be
a+ib=inf\{\varepsilon+i\eta|\hat{E}^{\hat{C}}_{\varepsilon}\in \mb\text{ and }\hat{E}^{\hat{B}}_{\eta}\in \mb\}\text{ iff }a+b=inf\{\varepsilon+\eta|\hat{E}^{\hat{C}}_{\varepsilon}\in \mb\text{ and }\hat{E}^{\hat{B}}_{\eta}\in \mb\}
\ee
\end{proof}
Let us now consider the following theorem:
\begin{Theorem}
Given a self-adjoint operator $\hat{A}=\hat{C}+\hat{D}$ then\footnote{Recall that generally for two functions we have $inf(f+g)\geq inf(f)+inf(g)$.} 
\be
f_{\hat{A}}=f_{\hat{C}}+f_{\hat{D}}
\ee

\end{Theorem}

\begin{proof}
\ba
f_{\hat{A}}(\mb)=f_{\hat{C}+\hat{B}}(\mb)&=&inf\{\sigma+\gamma|\hat{E}^{\hat{C}+\hat{B}}_{\sigma+\gamma}\in \mb\}\\
&=&inf\{\sigma+\gamma|\hat{E}^{\hat{C}}_{\sigma}\hat{E}^{\hat{B}}_{\gamma}\in \mb\}\\
&=&inf\{\sigma+\gamma|\hat{E}^{\hat{C}}_{\sigma}\in \mb\text{ and }\hat{E}^{\hat{B}}_{\gamma}\in \mb\}\\
&=&inf\{\sigma|\hat{E}^{\hat{C}}_{\sigma}\in \mb\}+inf\{\gamma|\hat{E}^{\hat{B}}_{\eta}\in \mb\}
\ea

\end{proof}
This follows since for any element $\varepsilon$ in the spectrum of $\hat{C}$ and any element $\eta$ in the spectrum of $\hat{B}$, such that $\hat{E}^{\hat{C}}_{\varepsilon}\in \mb$ and $\hat{E}^{\hat{B}}_{\eta}\in \mb$, the combination $\varepsilon+\eta$ will belong to the spectrum of $\hat{C}+\hat{B}$ and, consequently to the spectrum of $\hat{A}$. Recall also that we are now in $\Rl$.

Putting the results of the two theorems together we obtain
\be
f_{\hat{A}}(\mb)=inf\{\gamma|\hat{E}^{\hat{A}}_{\gamma}\in \mathcal{B}\}\text{ iff }(f_{\hat{C}}+f_{\hat{B}})(\mb)=inf\{\varepsilon|\hat{E}^{\hat{C}}_{\varepsilon}\in \mb\}+inf\{\eta|\hat{E}^{\hat{B}}_{\eta}\in \mb\}
\ee 
This converges to the following theorem:
\begin{Theorem}\label{the:sum}
Given a normal operator $\hat{A}=\hat{C}+i\hat{B}$ then
\be
f_A= f_C+if_B
\ee

\end{Theorem}
\begin{proof}
For any $\mb$ we get 
\be
a+ib=f_{\hat{A}}(\mb)=inf\{\gamma|\hat{E}^{\hat{A}}_{\gamma}\in \mathcal{B}\}
\ee 
but from above we know that this is the case iff 
\be
(f_{\hat{C}}+f_{\hat{B}})(\mb)=inf\{\varepsilon|\hat{E}^{\hat{C}}_{\varepsilon}\in \mb\}+inf\{\eta|\hat{E}^{\hat{B}}_{\eta}\in \mb\}=a+b
\ee
Thus 
\be
inf\{\varepsilon|\hat{E}^{\hat{C}}_{\varepsilon}\in \mb\}+(i) inf\{\eta|\hat{E}^{\hat{B}}_{\eta}\in \mb\}=a+ib
\ee
\end{proof}
Thus it would seem that $f_{\hat{A}}\equiv +_{\Cl} \circ \langle f_C,f_B\rangle\circ \bigtriangleup$. So for each $\mb\in \mathcal{O}(p(V))$ we have that
\be
f_{\hat{A}}\equiv +_{\Cl} \circ \langle f_C,f_B\rangle\circ \bigtriangleup(\mb)
\ee 
Can this then be generalised to the entire set of normal-observable functions and observable functions?

Denoting the set of all observable functions as $Ob$ and the set of all normal-observable functions as $On$ we define the map
\ba
h:Ob\times Ob&\rightarrow & On\\
(f_{\hat{C}},f_{\hat{B}})&\mapsto&h(f_{\hat{C}},f_{\hat{B}}):=f_{\hat{C}}+if_{\hat{B}}
\ea
where $\hat{A}=\hat{C}+i\hat{B}$ and $f_{\hat{A}}=f_{\hat{C}}+if_{\hat{B}}\equiv +_{\Cl} \circ \langle f_C,f_B\rangle\circ \bigtriangleup$.\\

The map $h$ is an isomorphism:
\begin{enumerate}
\item [i)] {\it 1:1}. Given $\hat{A}=\hat{C}+i\hat{B}$ and $\hat{D}=\hat{C}^{'}+\hat{D}^{'}$, if we assume that $h(f_{\hat{A}})=h(f_{\hat{D}})$ then for any $\mb$, $h(f_{\hat{A}})(\mb)=h(f_{\hat{D}})(\mb)$. Therefore $f_{\hat{C}}(\mb)+if_{\hat{B}}(\mb)=f_{\hat{C}^{'}}(\mb)+if_{\hat{B}^{'}}(\mb)=a+ib$. Thus $f_{\hat{C}}(\mb)=f_{\hat{C}^{'}}(\mb) $ and $f_{\hat{B}}(\mb)=f_{\hat{B}^{'}}(\mb)$. Since we are considering all $\mb$ then $f_{\hat{C}}=f_{\hat{C}^{'}}$ and $f_{\hat{B}}=f_{\hat{B}^{'}}$.
\item [ii)] {\it Onto}. This follows from theorem \ref{the:sum}.
\item [iii)] {\it Inverse}. The inverse would be
\ba
j:On&\rightarrow &Ob\times Ob\\
f_{\hat{A}}&\mapsto&j(f_{\hat{A}}):=(f_{\hat{C}},f_{\hat{B}})
\ea
where $\hat{A}=\hat{C}+i\hat{B}$.

\end{enumerate}
So, to each pair of normal-observable function one can uniquely associate a pair of observable functions. This is the lattice theoretical analogue of the fact that the spectrum of normal operators is defined in terms of the spectrum of the self-adjoint operators comprising it.

\section{Interpreting the Observable Functions for Normal Operators}
What does the above analysis tell us about the daseinisation of normal operators?
In particular, can we reproduce the analysis done in \cite{Doering2008} for the normal-observable functions and give a physical interpretation to the normal functions?

Since for each normal operator $\hat{A}$ we have that
\be
f_{\hat{A}}:\mathcal{Q}(\mathcal{P}(V))\rightarrow Sp(\hat{A})
\ee
and since the Stone spectrum $\mathcal{Q}(\mathcal{P}(V))$ is isomorphic to the Gel'fand spectrum when the algebra $V$ is abelian, it follows that each map $f_{\hat{A}}$ can be seen as generalisation of the Gel'fand transform of $\hat{A}$.
Given this, we would also like to interpret the map $f_{\delta^o(\hat{A})_V}$ as the generalised Gel'fand transform of $\delta^o(\hat{A})$.
In order to do this we need to reproduce all the calculations done in \cite{Doering2008}, but as applied to normal operators.

We first of all need to introduce the notion of a cone.
\begin{Definition}
Given a filter $F$ of a lattice $L$, a cone over $F$ in $L$ is the smallest filter in $L$ that contains $F$:
\be
C_{L}(F):=\{b\in L|\exists a\in F: a\leq b\}
\ee
This is basically an upper set of $F$ in $L$. 
\end{Definition}

We then want to show the validity of the following theorem:
\begin{Theorem}
Given a normal operator $\hat{A}$ in von Neumann algebra $N$ (not necessarily abelian) then for all von Neumann sub-algebras $S\subseteq N$ and filters $\mb\in\mathcal{O}(P(S))$ we have 
\be
f_{\delta^o(\hat{A})_S}(\mb)=f_{\hat{A}}(C_N(\mb))
\ee

\end{Theorem}
The proof is identical to the one in \cite{Doering2008} but now with the difference that the infimum is taken in the complex numbers ordered by ordering defined in \ref{def:order}. 
Given the normal operator $\hat{A}=\hat{C}+i\hat{B}$, since $\delta^o(\hat{A})\leq\delta^o(\hat{C})+i\delta^o(\hat{B})$ then 
\be
f_{\delta^o(\hat{A})_S}(\mb)\leq f_{\delta^o(\hat{C})_S}(\mb)+if_{\delta^o(\hat{B})_S}(\mb)
\ee
where $f_{\delta^o(\hat{C})_S}$ and $f_{\delta^o(\hat{B})_S}$ are the observable functions for the self-adjoint operators $\hat{C}$ and $\hat{B}$.
By considering $N=\mb(\mh)$ we obtain 
\be
f_{\delta^o(\hat{A})_V}(\mb)=f_{\hat{A}}(C_{\mb(\mh)}(\mb))
\ee
for all stages $V\in \mv(\mh)$ and all filters $\mb\in \mathcal{O}(P(V))$.

Combining all the results obtained so far we can write, as done in \cite{Doering2008}, the Gel'fand transform of the daseinsed normal operators in terms of observable functions of the non-daseinised normal operator:
\be
\overline{\delta^o(\hat{A})_V}(\lambda)=\langle\lambda, \delta^o(\hat{A})_V\rangle=f_{\delta^o(\hat{A})_V}(\mb_{\lambda})=f_{\hat{A}}(C_{\mb(\mh)}(\mb_{\lambda}))
\ee
where $\mb_{\lambda}$ is the ultrafilter associated with $\lambda$ as defined in equation \ref{ali:stone}.

We can define the antonymous functions for normal operators analogously to the definitions for self-adjoint operators:
\begin{Definition}
An antonymous function for the normal operator $\hat{A}$ is the function
\ba
g_{\hat{A}}:\mathcal{O}(P(V))&\rightarrow& \Cl\\
\mb&\mapsto&sup\{\lambda\in\Cl\hat{1}-\hat{E}^{\hat{A}}_{\lambda}\in \mb\}
\ea
\end{Definition}
\noindent It is easy to show that $im \left(g_{\hat{A}}\right)=sp(\hat{A})$.

It is worth pointing out that in a recent paper \cite{dewitt}, the authors show how self-adjoint operators in standard quantum theory can be represented by certain real valued functions called $q$-observables. In particular, given a von Neumann algebra $\mn$, a $q$-observable is a join-preserving function $o:P(\mn)\rightarrow\overline{\Rl}$\footnote{Here $\overline{\Rl}$ represents the extended reals: $\overline{\Rl}=\{-\infty\}\cup\Rl\cup\{\infty\}$ and $P(\mn)$ the lattice of projection operators in the algebra $\mn$.} which satisfy certain properties. The then show that there exists a bijective correspondence between the set $QO(P(\mn), \overline{\Rl})$ of $q$-observables and the set $SA(\mn)$ of self-adjoint operators affiliated\footnote{A self adjoint operator is said to be affiliated with $\mn$ if all the projection operators of its spectral decomposition lie in $\mn$. If $\hat{A}$ is bounded then $\hat{A}\in \mn_{sa}$, if $\hat{A}$ is unbounded then it is affiliated with $\mn$} with $\mn$. Such a correspondence is given in terms of an adjunction relation: each $o\in QO(P(\mn), \overline{\Rl})$ has as a right adjoint an extended\footnote{Here an extended spectral family is simply a spectral family whose definition was extended to $\overline{\Rl}$.}, right continuous, spectral family $\hat{E}^{o}\in SF(\overline{\Rl}, P(\mn))$ and conversely any $\hat{E}\in SF(\overline{\Rl}, P(\mn))$ has a left adjoint $o^{\hat{E}}\in QO(P(\mn), \overline{\Rl})$. \\
The authors then proceeded in showing that for $\mathcal{M}\subseteq\mn$ then $o^{\hat{A}}|_{P(\mathcal{M})}=o^{\delta^o(\hat{A})_{\mathcal{M}}}$, such that $o^{\delta^o(\hat{A})_{\mathcal{M}}}\dashv(\delta^i(\hat{E}^{\hat{A}}_r)_{\mathcal{M}})_{r\in \overline{R}}$.\\
A similar relation holds for the newly defined $q$-antonymous functions and outer inner daseinisation.

With this analysis, a much deeper mathematical understanding of daseinisation of self adjoint operators and their representation via maps from the spectral presheaf and the (real) quantity value object is obtained (see discussion in \cite{dewitt}).\\
Since the tools utilised in that paper can all be extended to the situation of normal operators we assume that a similar analysis can be done for normal operators. This would be a very interesting endeavour since it will make the mathematical significance of deaseinisation of normal operators much more clear. However, this analysis is left for future work.

\subsection{Physical Interpretation of the Arrow $\breve{\delta}(\hat{A}):\us\rightarrow\Cl^{\leftrightarrow}$ }
Given what has been said in the above sections it is clear that the arrow $\breve{\delta}(\hat{A}):\us\rightarrow\Cl^{\leftrightarrow}$ has the same exact physical interpretation as did $\breve{\delta}(\hat{A}):\us\rightarrow\Rl^{\leftrightarrow}$ for a self adjoint operator $\hat{A}$. Namely, given a state $|\psi\rangle$ then expectation value of the normal operator $\hat{A}$ is
\be
\langle\psi|\hat{A}|\psi\rangle=\int^{||\hat{A}||}_{-||\hat{A}||}\lambda d\langle\psi\hat{E}^{\hat{A}}_{\lambda}|\psi\rangle
\ee
Now, given the observable and antonymous functions defined above, we can re-write those expressions as
\be
\langle\psi|\hat{A}|\psi\rangle=\int^{f_{\hat{A}}(T^{|\psi\rangle})}_{f_{\hat{A}}(T^{|\psi\rangle})}\lambda d\langle\psi\hat{E}^{\hat{A}}_{\lambda}|\psi\rangle
\ee
where 
\be
T^{|\psi\rangle}:=\{\hat{\alpha}\in P(\mh)|\hat{\alpha}\geq|\psi\rangle\langle\psi|\}
\ee
is a maximal filter in $P(\mh)$. Thus we can see how $f_{\hat{A}}(T^{|\psi\rangle})$ represents the maximal value $\hat{A}$ can have, while $g_{\hat{A}}(T^{|\psi\rangle})$ would represent the minimum, i.e.
\be
g_{\hat{A}}(T^{|\psi\rangle})<\langle\psi|\hat{A}|\psi\rangle<f_{\hat{A}}(T^{|\psi\rangle})
\ee
Clearly if $|\psi\rangle$ is an eigenstate of $\hat{A}$, then the above inequalities all become equalities
\be
g_{\hat{A}}(T^{|\psi\rangle})=\langle\psi|\hat{A}|\psi\rangle=f_{\hat{A}}(T^{|\psi\rangle})
\ee

\section{Topos Notion of a one Parameter Group}
Since we now have the topos definitions of both the real and complex quantity value objects,
we can define the topos notion of a one parameter group with the parameter taking values either in the real value object or in the complex value object.
We will start with the former. 
Let us consider a one parameter group $\{\alpha(t)|t\in\Rl\}$ which defines an automorphisms of $\mh$.
We would now like to internalise such an object, i.e. to define the topos analogue of the automorphisms group $H=\{\alpha(t)|t\in \Rl\}$.
We know that for each element in this group we obtain the induced geometric morphisms
\ba
\alpha(t)^*:Sh(\mv(\mathcal{N}))&\rightarrow& Sh(\mv(\mathcal{N}))\\
\ps{S}&\mapsto&\alpha(t)^*\ps{S}
\ea
such that $\alpha_{\rho}(t)^*\ps{S}(V):=\ps{S}(\alpha_{\rho}(t)V)$.

Such an action, however, gives rise to twisted presheaves. To solve this problem we need to apply the methods defined in \cite{Flori2011} and use as the new base category the category $\mv_f(\mh)$ which is fixed, i.e. we do not allow any group to act on it. In Section \ref{sec:stone} we describe in more details how sheaves on the new category $\mv_f(\mh)$ are defined. 

We now define the internal group $\ps{H}$ over the new base category $\mv_f(\mh)$ as follows:
\begin{Definition}
The internal group $\ps{H}$ is the presheaf defined on
\begin{enumerate}
\item Objects: for each $V\in \mv_f(\mathcal{N})$ we obtain $\ps{H}_V=H$.
\item Morphisms: These are simply the identity maps. 
\end{enumerate}
\end{Definition}
\noindent It is straightforward to see that $\Gamma(\ps{H})=H$. 

We now would like to define the group $\ps{H}$ as a one parameter group of transformations, with parameter taking values in the quantity value object $\ps{\Rl}^{\leftrightarrow }$. \\
Generally, a one parameter group of transformations $\{\alpha(t)|\forall t\in \Rl\}$ is a representation of the additive abelian group $(\Rl,+)$. However, 
as shown in \cite{Doering2008}, $\ps{\Rl}^{\leftrightarrow }$ is only a commutative monoid, not an abelian group since, although addition ($+:\ps{\Rl}\times\ps{\Rl}\rightarrow \ps{\Rl}$) is well defined, subtraction is not. Fortunately, this difficulty is not insurmountable.

In order to extent a semigroup with unit to a full group, one strategy to use is the well known Grothendieck k-Construction already mentioned in \cite{Doering2008}. Such a construction is defined as follows:
\begin{Definition}
A group completion of a monoind $M$ is an abelian group $k(M)$ together
with a monoid map $\theta : M\rightarrow k(M)$ which is universal. Therefore, given
any monoid morphism $\phi: M\rightarrow G$, where $G$ is an abelian group, there exists
a unique group morphism $\phi^{'}: k(M)\rightarrow G$ such that $\phi$ factors through $\phi^{'}$,
i.e., the following diagram commutes
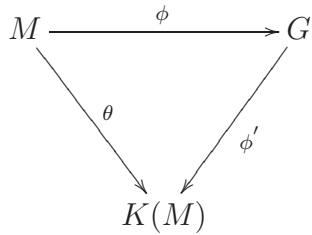
\[\xymatrix{
M\ar[rr]^{\phi}\ar[ddr]^{\theta}&&G\ar[ddl]^{\phi^{'}}\\
&&\\
&K(M)&\\
}\]

\end{Definition}
It is easy to show that any $k(M)$ is unique up to isomorphisms.
As showed in \cite{Doering2008} the construction of $k(M)$ is via an equivalence class.
This is because what is missing is the inverse (subtraction) operation, however, given two elements $(a, b)\in M\times M$, if we think of them as meaning $a-b$, then we notice that $a-b=c-d$ iff $a+d=c+b$. Thus one defines an equivalence relation on $M\times M$ as follows:
\be\label{equ:equi}
(a,b)\simeq(c,d)\text{ iff }\exists e\in M\text{ such that }a + d + e = b + c + e
\ee
Using this definition of equivalence we then equate $k(M)$ to precisely such a collection of equivalence classes where, again, each of them should be thought of as representing the subtraction of the two terms involved. This leads to the following definition:
\begin{Definition}
 The Grothendieck completion of an abelian monoid $M$ is the pair $(k(M),\theta)$ defined as follows:
 \begin{enumerate}
\item  $k(M)$ is the set of equivalence classes $[a, b]$, where the equivalence relation
is defined in \ref{equ:equi}. A group law on $k(M) $ is defined by
 \begin{enumerate}
\item [i)] $[a, b] + [c, d] := [a + c, b + d]$
\item [ii)] $0_{k(M)} := [0_M, 0_M]$
 \item  [iii)] $-[a, b] := [b, a]$
 \end{enumerate}
where $0_M$ is the unit in the abelian monoid $M$.
\item  The map $\theta : M\rightarrow  k(M) $ is defined by $\theta(a) := [a, 0]$ for all $a \in M$.

 \end{enumerate}

\end{Definition}

For the case at hand we then define the equivalence relation on $\ps{\Rl^{\leftrightarrow}}\times\ps{\Rl^{\leftrightarrow}}$ as follows:\\
for each context $V$ we have
\ba
& & \big((\mu_1, \nu_1),(\mu_2, \nu_2)\big)\equiv \big((\mu^{'}_1, \nu^{'}_1),(\mu^{'}_2, \nu^{'}_2)\big) \text{ iff } \exists (\mu, \nu)\in \ps{\Rl}_V \text{ such that }\\
& &(\mu_1, \nu_1) + (\mu^{'}_2, \nu^{'}_2) + (\mu, \nu) = (\mu_2, \nu_2) + (\mu^{'}_1, \nu^{'}_1) + (\mu, \nu)
\ea
Given such an equivalence class we can now define the object $k(\ps{\Rl^{\leftrightarrow}})$ as follows:
\begin{Definition}
The presheaf $k(\ps{\Rl^{\leftrightarrow}})\in \Sets^{\mv(\mh)}$ is defined on:
\begin{enumerate}
\item  Objects: for each $V\in \mv(\mh)$ we obtain 
\be
k(\ps{\Rl^{\leftrightarrow}})_V := \{[(\mu,\nu), (\mu^{'},\nu^{'})] | \mu, \mu^{'}  \in OP(\downarrow V, \Rl), \nu, \nu^{'}\in OR(\downarrow V, \Rl)\}
\ee
where $[(\mu,\nu), (\mu^{'},\nu^{'})]$ denotes the k-equivalence class of $(\mu,\nu), (\mu^{'},\nu^{'})$.
\item Morphisms: for each $i_{V^{'}V}:V^{'}\subseteq V$ we obtain the arrow $k(\ps{\Rl^{\leftrightarrow}})(i_{V^{'}V}): k(\ps{\Rl^{\leftrightarrow}})_V\rightarrow k(\ps{\Rl^{\leftrightarrow}})_{V^{'}}$ defined by
\be
k(\ps{\Rl^{\leftrightarrow}})(i_{V^{'}V})[(\mu,\nu), (\mu^{'},\nu^{'})]:=[(\mu_{|V^{'}},\nu_{|V^{'}}), (\mu^{'}_{|V^{'}},\nu^{'})_{|V^{'}}]\text{ for all }[(\mu,\nu), (\mu^{'},\nu^{'})]\in k(\ps{\Rl^{\leftrightarrow}})_V
\ee
\end{enumerate}

\end{Definition}
In this way we have obtained an abelian group object $k(\ps{\Rl^{\leftrightarrow}})$.
Is it now possible to define a one parameter group of automorphisms in terms of such an abelian group?

Due to the cumbersome notation we will use $k(\ps{\Rl}^{\geq})$ instead of the full $k(\ps{\Rl^{\leftrightarrow}})$. Here $k(\ps{\Rl}^{\geq})$ is the $k$-extention of the presheaf $k(\ps{\Rl}^{\geq})$ which is defined as follows:
\begin{Definition}
The presheaf $k(\ps{\Rl}^{\geq})\in \Sets^{\mv(\mh)}$ is defined on:
\begin{enumerate}
\item  Objects: for each $V\in \mv(\mh)$ we obtain 
\be
k(\ps{\Rl}^{\geq})_V := \{[\mu,\nu] | \mu, \nu \in OR(\downarrow V, \Rl)\}
\ee
where $[\mu,\nu]$ denotes the k-equivalence class of $(\mu,\nu)$.
\item Morphisms: for each $i_{V^{'}V}:V^{'}\subseteq V$ we obtain the arrow $k(\ps{\Rl}^{\geq})(i_{V^{'}V}): k(\ps{\Rl}^{\geq})_V\rightarrow k(\ps{\Rl}^{\geq})_{V^{'}}$ defined by
\be
k(\ps{\Rl}^{\geq})(i_{V^{'}V})[\mu,\nu]:=[\mu_{V^{'}},\nu_{V^{'}}]\text{ for all }[\mu,\nu]\in k(\ps{\Rl}^{\geq})_V
\ee
\end{enumerate}
\end{Definition}
This restriction causes no trouble since $k(\ps{\Rl}^{\geq})\subset k(\ps{\Rl^{\leftrightarrow}})$ and the results have an easy generalisation to $k(\ps{\Rl^{\leftrightarrow}})$. The advantage of using $k(\ps{\Rl}^{\geq})$ is that the notation is much more clear to understand.

\subsection{One Parameter Group Taking Values in $k(\ps{\Rl}^{\geq})$}
With the above discussion in mind we attempt the following definition: 
\begin{Definition}
The presheaf $\ps{K}\in Sets^{\mv_f(\mh)}$ is defined on 
\begin{enumerate}
\item Objects: for each context $V$ we define $\ps{K}_V:=\{\alpha_{[\mu, \nu]}|[\mu, \nu]\in k(\ps{\Rl}^{\geq})_V\}$.
\item Morphisms: given the inclusion $i_{V^{'}V}:V^{'}\subseteq V$ we define $\ps{K}(i_{V^{'}V}):\ps{K}_V\rightarrow \ps{K}_{V^{'}}$ as $\alpha_{[\mu, \nu]}\mapsto \alpha_{[\mu_{V^{'}}, \nu_{V^{'}}]}$.
\end{enumerate}

\end{Definition}
This is clearly a presheaf since given $V^{''}\xrightarrow{i}V^{'}\xrightarrow{j}V$, 
\ba
\ps{K}(i_{V^{'}V}\circ j_{V^{''}V}):\ps{K}_V&\rightarrow& \ps{K}_{V^{''}}\\
\alpha_{[\mu, \nu]}&\mapsto& \alpha_{[\mu_{V^{''}}, \nu_{V^{''}}]}
\ea
while
\ba
\ps{K}( j_{V^{''}V})\circ \ps{K}(i_{V^{'}V}):\ps{K}_V&\rightarrow& \ps{K}_{V^{'}}\rightarrow \ps{K}_{V^{''}}\\
\alpha_{[\mu, \nu]}&\mapsto& \alpha_{[\mu_{V^{'}}, \nu_{V^{'}}]}\rightarrow  \alpha_{[(\mu_{V^{'}})|_{V^{''}}, (\nu_{V^{'}})|_{V^{''}}]}=\alpha_{[\mu_{V^{''}}, \nu_{V^{''}}]}
\ea
The presheaf $\ps{K}$ can be turned into a group by defining the additive operation, for all $V\in \mv(\mh)$,  as follows: 
\ba
+_V:\ps{K}_V\times \ps{K}_V&\rightarrow &\ps{K}_V\\
(\alpha_{[\mu_1, \nu_1]}, \alpha_{[\mu_2, \nu_2]})&\mapsto&+_V(\alpha_{[\mu_1, \nu_1]}, \alpha_{[\mu_2, \nu_2]}):=\alpha_{[\mu_1+\mu_2, \nu_1+\nu_2]}
\ea
From now on we will denote $+_V(\alpha_{[\mu_1, \nu_1]}, \alpha_{[\mu_2, \nu_2]})$ as $\alpha_{[\mu_1, \nu_1]}\circ  \alpha_{[\mu_2, \nu_2]}$. The presheaf $\ps{K}$ is clearly closed under such additive structure. The inverse is defined as follows:

\begin{Definition}
For each $V\in\mv(\mathcal{N})$ we have
\ba
-_V:\ps{K}_V&\rightarrow &\ps{K}_V\\
\alpha_{[\mu_1, \nu_1]}&\mapsto&-_V(\alpha_{[\mu_1, \nu_1]}):=\alpha_{-[\mu_1, \nu_1]}=\alpha_{[\nu_1, \mu_1]}
\ea 
\end{Definition}
The unit element at each $V$ is defined as $\alpha_{[0, 0]}$ where each $(0)$ is the constant map with value $0$, hence it is both order reversing and order preserving. 
We now want to show that the group axioms hold. \\
{\it Associativity}\\
\be
(\alpha_{[\mu_1, \nu_1]}\circ  \alpha_{[\mu_2, \nu_2]})\circ \alpha_{[\mu_3, \nu_3]}=\alpha_{[\mu_1+\mu_2, \nu_1+\nu_2]}
\circ  \alpha_{[\mu_3, \nu_3]}=\alpha_{[\mu_1+\mu_2+\mu_3, \nu_1+\nu_2+\nu_3]}
\ee
On the other hand 
\be
\alpha_{[\mu_1, \nu_1]}\circ  (\alpha_{[\mu_2, \nu_2]}\circ \alpha_{[\mu_3, \nu_3]})=\alpha_{[\mu_1, \nu_1]}\circ \alpha_{[\mu_2+\mu_3, \nu_2+\nu_3]} =\alpha_{[\mu_1+\mu_2+\mu_3, \nu_1+\nu_2+\nu_3]}
\ee
{\it Identity Axiom}\\
\be
\alpha_{[\mu_1, \nu_1]}\circ  \alpha_{[0,0]}=\alpha_{[\mu_1+0, \nu_1+0]}= \alpha_{[0,0]}\circ \alpha_{[\mu_1, \nu_1]}
\ee
{\it Inverse Axiom}\\
\be
\alpha_{[\mu_1, \nu_1]}\circ \alpha_{-[\mu_1, \nu_1]}=\alpha_{[\mu_1-\mu_1, \nu_1-\nu_1]}=\alpha_{[0,0]}
\ee
From the above it follows that:
\begin{Proposition}
The group abelian $\ps{K}$ is a one parameter group with the parameter taking values in $k(\ps{\Rl}^{\geq})$.
\end{Proposition}
\begin{proof}
To prove the above theorem we need to define a continuous group homomorphism between $k(\ps{\Rl}^{\geq})$ and $\ps{K}$.  This is easy to do and, for each $V\in\mv(\mn)$, it is defined as follows
\ba
p_V:k(\ps{\Rl}^{\geq})_V&\rightarrow& \ps{K}_V\\
\;[\mu,\nu]&\mapsto&\alpha_{[\mu,\nu]}
\ea
Recalling that $k(\ps{\Rl}^{\geq})$ is equipped with the additive operation $+:k(\ps{\Rl}^{\geq})\times k(\ps{\Rl}^{\geq})\rightarrow k(\ps{\Rl}^{\geq})$, for each context $V\in \mv(\mathcal{N})$ we have
\ba
+_V:k(\ps{\Rl}^{\geq})_V\times k(\ps{\Rl}^{\geq})_V&\rightarrow& k(\ps{\Rl}^{\geq})_V\\
([\mu_1, \nu_1], [\mu_2, \nu_2])&\mapsto&+_V([\mu_1, \nu_1], [\mu_2, \nu_2]):=[\mu_1+\mu_2, \nu_1+\nu_2]
\ea
In order to prove continuity we need to equip both sheaves with a topology. We simply choose the discrete topology. This makes the above map $p_V$ continuous.
\end{proof}
We now consider the presheaf $\underline\Rl$ which is defined as follows:
\begin{Definition}
The presheaf $\ps{\Rl}\in Sets^{\mv(\mathcal{N})}$ is defined on 
\begin{enumerate}
\item Objects: for each $V\in \mv(\mathcal{N})$ we obtain $\ps{\Rl}_V=\Rl$.
\item Morphisms: given the inclusion $i:V^{'}\subseteq V$ the corresponding prehseaf map is simply the identity map.
\end{enumerate}
\end{Definition}
It is possible to embed $\underline\Rl$ in $\ps{\Rl}^{\geq}$ since each real number $r\in \ps{\Rl}_V$ can be identified with the constant function $c_{r, V}:\downarrow V\rightarrow\Rl$, which has constant value $r$ for all $V\in\mv(\mn)$. 
Such a function is trivially order-reversing, hence it is an element of $\ps{\Rl}^{\geq}$. 
Moreover, the global sections of $\ps{\Rl}$ are given by constant functions $r : \mv(\mathcal{N})\rightarrow \Rl$ which are also global sections of $\ps{\Rl}^{\geq}$, thus $\underline\Rl\subset \ps{\Rl}^{\geq}$. 
However $\ps{\Rl}^{\geq}$ can be seen as a sub-object of $k \left( \ps{\Rl}^{\geq} \right)$ by sending each $\mu\in \ps{\Rl}^{\geq}_V$ to $[\mu, 0]\in k \left( \ps{\Rl}^{\geq}_V \right)$, thus $\ps{\Rl}\subseteq k \left( \ps{\Rl}^{\geq} \right)$. 
Our claim is that $\ps{H}$ is isomorphic to the subgroup of $\ps{K}$ generated by $\ps{\Rl}\subset \ps{\Rl}^{\geq}$.
\begin{Theorem}
The group $\ps{H}$ is isomorphism to the one-parameter subgroup of $\ps{K}$ generated by $\ps{\Rl}$.
\end{Theorem}
\begin{proof}
We want to show that there exists a map
\be
f:\ps{\Rl}\rightarrow\ps{K}
\ee
such that for each context  $V$, $f_V$ is a continuous injective group homomorphism and $Im(f)\simeq \ps{H}\subset \ps{K}$.

We thus define, for each $V$ 
\ba
f_V:\ps{\Rl}_V&\rightarrow&\ps{K}_V\\
r&\mapsto&\alpha_{[c_{r,V},0]}
\ea
This is a well defined functor since for $V^{'}\subseteq V$ the following diagram commutes
\[\xymatrix{
\ps{\Rl}_V\ar[rr]^{f_V}\ar[dd]&&\ps{K}_V\ar[dd]\\
&&\\
\ps{\Rl}_{V^{'}}\ar[rr]_{f_{V^{'}}}&&\ps{K}_{V^{'}}\\
}\]
In fact we have
\be
\ps{K}(i_{V^{'}V})\circ f_V(r)=\ps{K}(i_{V^{'}V})\alpha_{[c_{r,V}, 0]}=\alpha_{[c_{r,V^{'}}, 0]}
\ee
while
\be
f_{V^{'}}\circ \ps{\Rl}(i_{V^{'}V})(r)=f_{V^{'}}(r)=\alpha_{[c_{r,V^{'}}, 0]}
\ee
Clearly $f$ is injective and continuous on the image.

We now need to check whether $f$ is a group homomorphism, i.e. we need to show that $f_V(r_1+r_2)=f_V(r_1)+f_V(r_2)$. We know that the left hand side is 
\be
f_V(r_1+r_2)=\alpha_{[c_{(r_1+r_2),V}, 0]}
\ee
However 
\ba
c_{(r_1+r_2),V}:\downarrow V&\rightarrow& \Rl\\
V^{'}&\mapsto&c_{(r_1+r_2),V}(V^{'})=r_1+r_2
\ea
while
\ba
+_V(c_{r_1, V}, c_{r_2, V}):\downarrow V&\rightarrow& \Rl\\
V^{'}&\mapsto&c_{r_1, V}(V^{'})+ c_{r_2, V}(V^{'})=r_1+r_2
\ea
Hence
\be
c_{(r_1+r_2),V}=c_{r_1, V}+ c_{r_2, V}
\ee
which implies that
\be
\alpha_{[c_{(r_1+r_2),V}, 0]}=\alpha_{[c_{r_1,V}+c_{r_2, V}, 0]}=\alpha_{[c_{r_1,V},0]}\circ \alpha_{[c_{r_2, V}, 0]}=f_V(r_1)+f_V(r_2)
\ee
We now want to show that $im(f)\simeq \ps{H}$. We therefore construct the map $i:im(f)\rightarrow\ps{H}$, such that for each $V$ we obtain
\ba
i_V:im(f)_V&\rightarrow&\ps{H}_V\\
\alpha_{[c_{r,V},0]}&\mapsto&i_V(\alpha_{[c_{r,V},0]}):=\alpha_{\rho}(c_{r,V}(V))
\ea
This is clearly an isomorphism.
We could have defined the map 
\ba
h_V:\ps{K}_V&\rightarrow&\ps{H}_V\\
\alpha_{[\mu, \nu]}&\mapsto&\alpha_{\rho}(\mu(V)+\nu(V))
\ea
but this would not have been 1:1.
\end{proof}

A real number $r\in\ps{\Rl}_V$ defines the pair $(c_{r,V} , c_{r,V} )$ given by of two copies of the constant function $c_{r, V}:\downarrow V\rightarrow\Rl$. Clearly such a function is both order-preserving and order-reversing, hence $(c_{r,V} , c_{r,V} )\in\ps{\Rl^{\leftrightarrow}}$. However $\ps{\Rl^{\leftrightarrow}}\subset k({\Rl}^{\leftrightarrow})$, therefore $\ps{\Rl}\subset k({\Rl}^{\leftrightarrow})$. Since all the results proved for $\ps{\Rl}^{\geq}$ hold for $\ps{\Rl^{\leftrightarrow}}$ but the constructions for $\ps{\Rl^{\leftrightarrow}}$ are more cumbersome, we will avoid reporting them here.
\subsection{One Parameter Group Taking Values in $k(\ps{\Cl}^{\geq})$}
We would now like to apply the same analysis but for the complex number object $\ps{\Cl^{\leftrightarrow}}$.
As before, we will consider the object $\ps{\Cl}^{\geq}$ (defined below) since it is more practical for notations.
All results will then translate in a simple way to $\ps{\Cl^{\leftrightarrow}}$. 
\begin{Definition}
The presheaf $\ps{\Cl}^{\geq}\in Sets^{\mv(\mathcal{N})}$ is defined on 
\begin{enumerate}
\item Objects: for each $V\in \mv(\mathcal{N})$ we obtain the set $\ps{\Cl}^{\geq}_V:=\{\mu|\text{ s.t. }\mu\in OR(\downarrow V,\Cl)\}$.
\item Morphisms: given $i_{V^{'}V}:V^{'}\subseteq V$ the presheaf morphism is defined by the restriction: $\ps{\Cl}^{\geq}_{V}\rightarrow \ps{\Cl}^{\geq}_{V^{'}}$; $\mu\mapsto\mu_{|V^{'}}$.
\end{enumerate}

\end{Definition}
We then define the $k$-extension $k(\ps{\Cl}^{\geq})$ as follows:
\begin{Definition}
The presheaf $k(\ps{\Cl}^{\geq})\in \Sets^{\mv(\mh)}$ is defined on:
\begin{enumerate}
\item  Objects: for each $V\in \mv(\mh)$ we obtain 
\be
k(\ps{\Cl}^{\geq})_V := \{[\mu,\nu] | \mu, \nu \in OR(\downarrow V, \Cl)\}
\ee
where $[\mu,\nu]$ denotes the k-equivalence class of $(\mu,\nu)$.
\item Morphisms: for each $i_{V^{'}V}:V^{'}\subseteq V$, we obtain the arrow $k(\ps{\Cl}^{\geq})(i_{V^{'}V}): k(\ps{\Cl}^{\geq})_V\rightarrow k(\ps{\Cl}^{\geq})_{V^{'}}$ defined by
\be
k(\ps{\Cl}^{\geq})(i_{V^{'}V})[\mu,\nu]:=[\mu_{V^{'}},\nu_{V^{'}}]\text{ for all }[\mu,\nu]\in k(\ps{\Cl}^{\geq})_V
\ee
\end{enumerate}

\end{Definition}
It is interesting to note how in $k(\ps{\Cl}^{\geq})$ it is possible to define complex conjugation. In particular we define, for each context $V$
\ba
*_V:k(\ps{\Cl}^{\geq})_V&\rightarrow& k(\ps{\Cl}^{\geq})_V\\
\;[\mu, \nu]&\mapsto&[\mu^*, \nu^*]
\ea
where $\mu^*(V):=(\mu(V))^*$.
Now if $(\mu, \nu)\simeq (\eta, \beta)$ then there exists an element $\gamma\in k(\ps{\Cl}^{\geq})$ such that $\mu+\beta+\gamma=\nu+\eta+\gamma$, which is defined for each $V^{'}\in\downarrow V$ as $\mu(V^{'})+\beta(V^{'})+\gamma(V^{'})=\nu(V^{'})+\eta(V^{'})+\gamma(V^{'})$, thus obtaining the equality of complex numbers $(a+b+c)+i(d+e+f)=(a_1+b_1+c)+i(d_1+e_1+f)$.
Applying the complex conjugation map we obtain $(\mu^*, \nu^*)\simeq (\eta^*, \beta^*)$ iff $\mu^*+\beta^*+\gamma^*=\nu^*+\eta^*+\gamma^*$ which, by applying the same reasoning, translates to $(a+b+c)-i(d+e+f)=(a_1+b_1+c)-i(d_1+e_1+f)$.
It follows that if $(\mu, \nu)\in [\mu, \nu]$ then $(\mu^*, \nu^*)\in [\mu^*, \nu^*]$. Thus $*_V$ is well defined.

We now want to show that $k(\Cl^{\geq})$ is a vector space over $\ps{R}$.
To this end we need to define multiplication with respect to an element in $\ps{\Rl}$.
We recall that each element $r\in \Rl$ is represented in $\ps{\Rl}$ as the global element $c_r\in\Gamma(\ps{\Rl})$ which, at each context $V\in\mv(\mh)$, defines a constant function $c_{r,V} :\downarrow V\rightarrow \Rl$.
We can then define multiplication with respect to such constant functions.
Thus, given a context $V\in \mv(\mh)$ we consider an element $[\mu,\nu]\in k(\ps{\Cl}^{\geq})_V$, and we define multiplication by $c_{r,V}$ as
\be\label{equ:real}
(c_{r,V} [\mu,\nu]) :=\begin{cases}[c_{r,V} \mu, c_{r,V}\nu] = [r\mu, r\nu] \text{ if } r \geq 0\\
-[c_{-r,V}\mu, c_{-r,V} \nu] = [-r\nu,-r\mu] \text{ if }  r < 0
\end{cases}
\ee
where for each $V^{'}\in \downarrow V$ the above is defined as
\be
(c_{r,V} [\mu,\nu])(V^{'})=c_{r,V}(V^{'})[\mu(V^{'}),\nu(V^{'})]=[r\mu(V^{'}),r\nu(V^{'})]=[r\mu,r\nu](V^{'})
\ee
for $r\geq0$, while for $r<0$ we have
\ba
(c_{r,V} [\mu,\nu])(V^{'})&=&c_{r,V}(V^{'})[\mu(V^{'}),\nu(V^{'})]=-|r|[\mu(V^{'}),\nu(V^{'})]=|r|[\nu(V^{'}),\mu(V^{'})]\\
&=&-[c_{-r,V}\mu, c_{-r,V} \nu] (V^{'})=[-r\nu,-r\mu](V^{'})\;.
\ea

Similarly we can also define multiplication with respect to a constant complex number. In fact, given a complex number $z=x+iy\in \Cl$ this represents a global element in $\Gamma(\ps{\Cl})$ such that, for each context $V\in \mv(\mh)$, we obtain the constant function $c_{z, V}:\downarrow V\rightarrow\Cl$. Thus, given an element $[\mu,\nu]\in k(\ps{\Cl}^{\geq})_V$ we define for each $V^{'}\in \downarrow V$, when  $x + iy \geq 0$
\be
(c_{z,V}[\mu,\nu])(V^{'})=c_{z,V}(V^{'})[\mu(V^{'}),\nu(V^{'})]=[z\mu(V^{'}),z\nu(V^{'})]=[z\mu,z\nu](V^{'})
\ee
On the other hand for $y=0$ and $x<0$, such that $x+iy<0$ we have
\be
(c_{z,V}[\mu,\nu])(V^{'})=c_{z,V}(V^{'})[\mu(V^{'}),\nu(V^{'})]=-[|z|\mu(V^{'}),|z|\nu(V^{'})]=[-z\nu,-z\mu](V^{'})
\ee
It is straight forward to see how this definition reduces to definition \eqref{equ:real}.\\
We now define the presheaf $\ps{Q}$ as follows:
\begin{Definition}
The presheaf $\ps{Q}\in\Sets^{\mv(\mathcal{N})}$ is defined on
\begin{enumerate}
\item Objects: for each $V\in\mv(\mathcal{N})$ we obtain the set $\ps{Q}_V:=\{\alpha([\mu, \nu])|[\mu, \nu]\in k(\ps{\Cl}^{\geq})_V\}$.
\item Morphisms: Given the inclusion map $V^{'}\subseteq V$ the corresponding presheaf map is $\ps{Q}(i_{V^{'}V}):\ps{Q}_V\rightarrow \ps{Q}_{V^{'}}$; $\alpha([\mu, \nu])\mapsto\alpha([\mu_{V^{'}}, \nu_{V^{'}}])$.
\end{enumerate}
\end{Definition}
 $\ps{Q}$ can be given a group structure in exactly the same way as was done for $\ps{K}$. It then follows that $\ps{Q}$ is the one parameter group defined via the group homomorphisms $h:k(\ps{\Cl}^{\geq})\rightarrow \ps{Q}$, which have components for each context
\ba
h_V:k(\ps{\Cl}^{\geq})_V&\rightarrow &\ps{Q}_V\\
\;[\mu, \nu]&\mapsto&\alpha([\mu, \nu])
\ea
\begin{Proposition}
The group $\ps{K}$ is a subgroup of $\ps{Q}$.
\end{Proposition}
\begin{proof}
$\ps{K}$ is the one-parameter subgroup generated by $k(\ps{\Rl}^{\geq})\subset k(\ps{\Cl}^{\geq})$. In fact we have the following continuous group homomorphisms for each $V$:
\ba
k(\ps{\Rl}^{\geq})_V&\rightarrow & k(\ps{\Cl}^{\geq})_V\rightarrow\ps{Q}_V\\
\;[\mu, \nu]&\mapsto&[\mu+i0, \nu+i0]\mapsto\alpha([\mu+i0, \nu+i0])=\alpha([\mu, \nu])
\ea
where again we are assuming the discrete topology.
\end{proof}
We now analyse the relation between $\ps{\Cl}$ and $k(\ps{\Cl}^{\geq})$. In particular, as for the real number object, we have that $\ps{\Cl}\subset k(\ps{\Cl}^{\geq})$. This inclusion is given by the following chain of inclusions for each $V$:
\ba
\gamma_V:\ps{\Cl}_V&\rightarrow&\ps{\Cl}^{\geq}_V\rightarrow k(\ps{\Cl}^{\geq})_V\\
t&\mapsto&c_{t, V}\rightarrow [c_{t, V}, 0]
\ea
The proof is straightforward.
As was done for the real valued number case we would like to define the topos analogue of the group $R_V:=\{\alpha_{\rho}(a+ib)|a+ib\in\Cl\}$, which takes values in $\Cl$.
We first construct the following presheaf:
\begin{Definition}
The presheaf $\ps{R}\in \Sets^{\mv_f(\mh)}$ is defined on:
\begin{enumerate} 
\item Objects: for each $V\in\mv(\mn)$ we obtain the set $\ps{R}_V:=\{\alpha_{\rho}(a+ib)|a+ib\in\Cl\}$.
\item Morphisms: for any map $i_{V^{'}V}:V^{'}\subseteq V$, $\ps{R}(i_{V^{'}V})$ is simply the identity.
\end{enumerate}
\end{Definition}
The fact that this presheaf is a group come from the fact that for each $V$, $\ps{R}_V$ is a group.
We would like to show that such a group object is a one parameter subgroup of $\ps{Q}$ taking its values in $\ps{\Cl}$.

To this end we construct the map $\phi:\ps{\Cl}\rightarrow \ps{Q}$, whose definition requires the factorisation via the map $\gamma$ above.
Thus, for each $V$, we have
\ba
\phi_V:\ps{\Cl}_V&\rightarrow& \ps{Q}_V\\
t&\mapsto&\phi_V(\gamma_V(t)):=\alpha([c_{t, V}, 0])
\ea
Clearly such a map is injective. We need to show that it is also an homomorphism. In particular we need to show that
\be
\phi_V(t_1+t_2)=\phi_V(t_1)\circ \phi_V(t_2)
\ee
By applying the definition we have
\be
\phi_V(t_1+t_2)=\alpha([c_{t_1+t_2, V}, 0])=\alpha([c_{t_1, V}, 0])\circ \alpha([c_{t_2, V}, 0])
\ee
where the last equation follows from the group laws in $\ps{Q}_V$.

On the other hand
\be
\phi_V(t_1)\circ \phi_V(t_2)=\alpha([c_{t_1, V}, 0])\circ \alpha([c_{t_2, V}, 0])
\ee
We thus obtain the one parameter subgroup of $\ps{Q}$ as the image of $\phi$, i.e. $im(\phi)\subset \ps{Q}$. 

We can then define the map $m:im(\phi)\rightarrow\ps{R}$ such that for each context $V$ we have
\ba
m_V:im(\phi)_V&\rightarrow& \ps{R}_V\\
\alpha([c_{t_2, V}, 0])&\mapsto&m_V(\alpha([c_{t_2, V}, 0])):=\alpha_{\rho}(c_{t_2, V}(V))
\ea
This is clearly an isomorphism.

\section{Stone's Theorem in the Language of Topos Theory}\label{sec:stone}
In the previous section we managed to define the topos analogue of the one parameter group of transformations. Since we are also able to define the topos analogue of self-adjoint operators, it is natural to ask whether it is possible to formulate Stone's theorem in the language of topos theory.
The ``standard'' definition of Stone's theorem is the following:
\begin{Theorem}
Every strongly continuous\footnote{Here strongly continuous means that, for any $\psi\in\mh$ and $t\rightarrow t_0$, then $ U_t(\psi)\rightarrow U(t_0)(\psi)$.}
 one-parameter group $\{U_t\}$, $(-\infty<t<\infty)$ of unitary transformations admits a spectral representation 
 \be
 U_t=\int_{\-\infty}^{\infty}e^{i\lambda t}d\hat{E}_{\lambda}
 \ee
 where $\{\hat{E}_{\lambda}\}$ is the spectral family such that\footnote{We will now introduce the following notations i) $\hat{A}_{\smile}\hat{B}$ indicates that $\hat{A}$ and $\hat{B}$ commute; ii) $\hat{A}_{\smile\smile}\hat{B}$ means that $\hat{A}$ commutes with $\hat{B}$ and any other operator which commutes with $\hat{B}$.} $(\hat{E}_{\lambda})_{\smile\smile}\{U_t\}$.
 \end{Theorem}
\noindent Equivalently one can write $U_t=e^{i\lambda t\hat{A}}$ for the self adjoint operator 
\be
\hat{A}=\int_{\-\infty}^{\infty}\lambda\hat{E}_{\lambda}
\ee
We are now interested in translating the above theorem into the topos language, that is, we are interested in finding a correspondence between self-adjoint operators $\breve{\delta}(\hat{A})$ and unitary one parameter groups.

First of all we need to specify what a unitary one parameter group is in a topos. We already have the definition of a one parameter subgroup, thus all we need to do is to add the property of unitarity. We thus consider the one parameter group $Q:=\{\alpha(t)|t\in \Rl \;\&\; \alpha(t)\alpha(-t)=1\}$ of transformations on $\mh$. These transformations can be extended to functors:
\ba
\alpha(t):\mv(\mh)&\rightarrow&\mv(\mh)\\
V&\mapsto&l_{\alpha(t)}V:=\{\alpha(t)\hat{A}\alpha(-t)|\hat{A}\in V\}
\ea

We then define the associated presheaf $\ps{Q}$ which has, as objects, for each $V\in\mv(\mh)$ the entire group $\ps{Q}_V=Q$.
The maps are simply the identity maps. The group $\ps{Q}$ represents a one parameter sub-group of $\ps{K}$ of unitary transformations. 
The proof is similar as the proof given above for the sub-group $\ps{H}$ while the unitarity is derived directly from $Q$. 
We should also add the property of strong continuity which, in terms of operators, can be stated as follows: for any $\hat{A}$ and $t\rightarrow t_0$ then $\alpha(t)\hat{A}\alpha(-t)\rightarrow \alpha(t_0)\hat{A}\alpha(-t_0)$. 

Given such a strongly continuous one-parameter sub-group of transformations we want to somehow define a unique self-adjoint operator associated to it and, vice versa, given a self adjoint operator we want to associate to it a unique strongly continuous one-parameter sub-group of transformations. We will start from the latter. Since we will be employing group transformations we need to work with the sheaves $\breve{\us}$ and $\breve{\Rl}^{\leftrightarrow}$ which are defined using the method introduced in \cite{Flori2011}. In particular, given the presheaf $\ps{K}$ we define the presheaf $\ps{K/K_F}$ as follows
\begin{Definition}
The presheaf $\ps{K/K_F}\in Sets^{\mv_f(\mh)}$\footnote{Recall that $\mv_f(\mh)$ is the poset $\mv(\mh)$ but were the group is not allowed to act.} is defined on
\begin{enumerate}
\item Objects: for each $V\in \mv_f(\mh)$ we obtain the set $K/K_{FV}=\{[g]_V|g\sim g_1 \text{ iff } hg_1=g \text{ for } h\in K_{FV}\}$ where $K_{FV}$ is the fixed point group of $V$.
\item Morphisms: given a morphism $i_{V^{'}V}:V^{'}\subseteq V$ the corresponding presheaf morphism is the map $K/K_{FV}\rightarrow K/K_{FV^{'}}$, defined as the bundle map of the bundle $K_{FV^{'}}/K_{FV}\rightarrow K/K_{FV}\rightarrow K/K_{FV^{'}}$.
\end{enumerate} 
\end{Definition}
From the above presheaf we obtain the associated \'etale bundle $p:\Lambda(\ps{K/K_{F}})\rightarrow \mv_f(\mathcal{N})$, whose bundle space $\Lambda(\ps{K/K_{F}})$ can be given a poset structure as follows:
\begin{Definition}
Given two elements $[g]_{V^{'}}\in K/K_{FV^{'}} $, $[g]_V\in K/K_{FV}$ we define the partial ordering by
\be
[g]_{V^{'}}\leq [g]_{V} \text{ iff } p([g]_{V^{'}})\subseteq  p([g]_V)\text{ and } [g]_V\subseteq  [g]_{V^{'}}
\ee
\end{Definition}
Next one defines the functor $I:Sh(\mv(\mathcal{H}))\rightarrow Sh(\Lambda(\ps{K/K_F})$ as follows:
\begin{Theorem}
The map $I:Sh(\mv(\mh))\rightarrow Sh(\Lambda( \ps{K/K_F}))$ is a
functor defined on
\begin{enumerate}
\item [(i)] Objects: $\big(I(\underline{A})\big)_{[g]_V}:=\underline{A}_{l_g(V)}=\Big((l_g)^*(\underline{A})\Big)(V)$. If $[g]_{V^{'}}\leq [g]_{V}$, then
$$(I\underline{A}(i_{[g]_{V^{'}},[g]_V})):=\underline{A}_{l_g(V),l_g(V^{'})}:\underline{A}_{l_g(V)}\rightarrow \underline{A}_{l_g(V^{'})}$$
where $V=p([g]_V)$ and $V^{'}=p([g]_{V^{'}})$.
\item [(ii)] Morphisms: given a morphism $f:\underline{A}\rightarrow\underline{B}$ in $Sh(\mv(\mh))$ we then define the corresponding morphism in $Sh(\Lambda (\ps{K/K_F}))$ as
\ba
I(f)_{[g]_V}:I(\underline{A})_{[g]_V}&\rightarrow& I(\underline{B})_{[g]_V}\\
f_{[g]_V}:\underline{A}_{l_g(p([g]_V))}&\rightarrow
&\underline{B}_{l_g(p([g]_V))} \ea
\end{enumerate}
\end{Theorem}
Such a functor was already defined in \cite{Flori2011} where it was shown to be a functor. We are now able to map all the sheaves in $Sh(\mv(\mathcal{H}))$ to sheaves in $Sh(\Lambda(\ps{K/K_F})$. By then applying the functor $p!:Sh(\Lambda(\ps{K/K_F})\rightarrow Sh(\mv_f(\mathcal{H}))$ we finally obtain sheaves on our fixed category $\mv_f(\mathcal{H})$. The advantage of this construction is that now the actions of both groups $\ps{K}$ and $\ps{H}$ do not induce twisted presheaves.
In this context, self adjoint operators are defined as
\be
\overline{\delta}(\hat{A}):\coprod_{g\in K/K_{FV}}\us_{l_gV}\rightarrow\coprod_{g\in K/K_{FV}}\ps{\Rl}_{l_gV}
\ee
i.e. as co-products of the originally defined self-adjoint operators $\breve{\delta}(\hat{A})$. For details see \cite{Flori2011}. 
By denoting all the maps $\breve{\us}\rightarrow\breve{\ps{\Rl}}^{\leftrightarrow}$ by $(\breve{\ps{\Rl}}^{\leftrightarrow})^{\breve{\us}}$ we can then define the sub-object $\ps{Ob}\subseteq (\breve{\ps{\Rl}}^{\leftrightarrow})^{\breve{\us}}$ of observables, i.e. all maps $\overline{\delta}(\hat{A}):\breve{\us}\rightarrow\breve{\ps{\Rl}}^{\leftrightarrow}$ associated to self adjoint operators $\hat{A}$. Next we define the collection of strongly continuous unitary subgroups of $\ps{K}$ which we denote $Sub_u(\ps{K})$. Given these ingredients we attempt the partial definition of Stone's theorem
\begin{Theorem}
The map $f:\ps{Ob}\rightarrow Sub_u(\ps{K})$ defined for each $V\in\mv(\mh)$ as 
\ba
f_V:\ps{Ob}_V&\rightarrow& Sub_u(\ps{K})_V\\
\overline{\delta}(\hat{A})_{|\downarrow V}&\mapsto& \ps{Q}_V^{\hat{A}}
\ea
where $\ps{Q}_V^{\hat{A}}:=\{e^{it\hat{A}}|t\in \Rl\}$ is injective.
\end{Theorem}
The proof of injectivity is trivial, thus what remains to show is that indeed $\ps{Q}_V^{\hat{A}}$, as defined above, is a strongly continuous unitary subgroup of $\ps{K}$. The proof is again similar to the one done for the sub-group $\ps{H}$, so we will not report it here. 
We now come to the more interesting part of Stone's theorem, namely showing that any strongly continuous unitary subgroups of $\ps{K}$ uniquely determines a self-adjoint operator. To this end we first construct the map $g:Sub_u(\ps{K})\rightarrow \ps{Ob}$ such that, for each $V\in\mv(\mh)$ 
\ba
g_V:Sub_u(\ps{K})_V&\rightarrow& \ps{Ob}_V\\
\ps{Q}_V&\mapsto& g_V(\ps{Q}_V):=\overline{\delta}(\hat{A}^Q)_{|\downarrow V}
\ea
The self-adjoint operator $ \overline{\delta}(\hat{A}^Q)_{|\downarrow V}$ is defined by the following properties:
\begin{enumerate}
\item [a)] For all $\alpha(t)\in\ps{Q}_V$, the diagram
\[\xymatrix{ 
\breve{\us}_V\ar[rr]^{\overline{\delta}(\hat{A}^Q)_V}\ar[dd]_{\alpha(t)^*}&&\ps{\Rl^{\leftrightarrow}}_V\\
&&\\
\breve{\us}_{V}\ar[rruu]_{\overline{\delta}(\hat{A}^Q)_{V}}&&
}
\]
commutes. What this means is that, given an element $\lambda\in \us_{V}\in\coprod_{g\in K/K_{FV}}\us_{l_gV}$ we require
\be
(\overline{\delta}(\hat{A}^Q)_V\circ \alpha(t)^*)\lambda=(\overline{\delta}(\hat{A}^Q)_V)l_{\alpha(t)}\lambda=\breve{\delta}(\hat{A}^Q)_{l_{\alpha(t)}V}(l_{\alpha(t)}\lambda)
\ee
to be equal to 
\be
\overline{\delta}(\hat{A}^Q)_{V}(\lambda)
\ee
We can generalise such a condition for all elements of $Q$ at once by requiring that the following diagram commutes:
\[\xymatrix{ 
Q\times \breve{\us}_V\ar[rr]^{\overline{\delta}(\hat{A}^Q)_V}\ar[dd]_{pr_2}&&\ps{\Rl^{\leftrightarrow}}_V\\
&&\\
\breve{\us}_{V}\ar[rruu]_{\overline{\delta}(\hat{A}^Q)_{V}}&&
}
\]

\item [b)] For any $\hat{B}$ such that $\hat{B}_{\smile}\hat{A}^Q$ then 
\[\xymatrix{ 
Q\times \breve{\us}_V\ar[rr]^{\overline{\delta}(\hat{B})_V}\ar[dd]_{pr_2}&&\ps{\Rl^{\leftrightarrow}}_V\\
&&\\
\breve{\us}_{V}\ar[rruu]_{\overline{\delta}(\hat{B})_{V}}&&
}
\]
\end{enumerate}
The correspondence between strongly continuous unitary groups and self adjoint operators is given by they following theorem:
\begin{Theorem}\label{the:stone}
The map $g$ is injective.
\end{Theorem}
\noindent Before proving the theorem let us first analyse in more details what the two conditions a) and b) imply. To this end we introduce the following corollary
\begin{Corollary}
Given a self adjoint operator $\hat{A}$ with spectral projection $\{\hat{E}_{\lambda}\}$ then 
\be
 \{\hat{E}_{\lambda}\}_{\smile}\{\alpha(t)\}\Leftrightarrow \hat{A}\text{  satisfies condition } a)
\ee
On the other hand
\be
 \{\hat{E}_{\lambda}\}_{\smile\smile}\{\alpha(t)\}\Leftrightarrow \hat{A}\text{  satisfies condition } a)\;\&\; b)
\ee
\end{Corollary}
We now prove the above corollary.
\begin{proof}
We assume that $\hat{A}$ satisfies condition a), i.e. $(\overline{\delta}(\hat{A})_V\circ \alpha(t)^*)\lambda=\overline{\delta}(\hat{A})_{V}(\lambda)
$. We recall that 
\be
(\overline{\delta}(\hat{A})_V\circ \alpha(t)^*)\lambda(\cdot)=\breve{\delta}(\hat{A})_{l_{\alpha(t)}V}(l_{\alpha(t)}\lambda)(\cdot)=(\breve{\delta}^i(\hat{A})_{l_{\alpha(t)}V}(l_{\alpha(t)}\lambda),\breve{\delta}^o(\hat{A})_{l_{\alpha(t)}V}(l_{\alpha(t)}\lambda))(\cdot)
\ee
such that, for any $V'\in \downarrow V$ we have
\be
\breve{\delta}^i(\hat{A})_{l_{\alpha(t)}V}(l_{\alpha(t)}\lambda)(V')=\breve{\delta}^i(\hat{A})_{l_{\alpha(t)}V'}(l_{\alpha(t)}\lambda)=l_{\alpha(t)}\lambda(\delta^i(\hat{A})_{l_{\alpha(t)}V'})=\lambda \delta(\alpha(-t)(\hat{A})\alpha(t))_{V'}
\ee
Therefore for condition a) to be satisfied it implies that $\{\hat{A}_{\lambda}\}_{\smile}\{\alpha(t)\}$. Since $\{\hat{E}_{\lambda}\}_{\smile\smile}\{\hat{A}\}$ it follows that $\{\hat{E}_{\lambda}\}_{\smile}\{\alpha(t)\}$. The converse is trivial to prove.
If we now assume that $\hat{A}$ also satisfies condition b) we then have that, given any other operator $\hat{B}_{\smile}\hat{A}$ then
\be
\lambda \delta(\alpha(-t)(\hat{B})\alpha(t))_{V'}=\lambda \delta(\hat{B})_{V'}
\ee
which implies that $\hat{B}_{\smile}\{\alpha(t)\}$. But since $\hat{A}_{\smile} \hat{B}$ and $\hat{A}_{\smile}\{\alpha(t)\}$ it follows that $\hat{A}_{\smile\smile}\{\alpha(t)\}$ hence $\{\hat{E}_{\lambda}\}_{\smile\smile}\{\alpha(t)\}$. Again the converse is easy to prove. 
\end{proof}
\noindent We are now ready to prove theorem \ref{the:stone}. 
\begin{proof}
We want to show that the map $g$ is injective. In particular, given two strongly continuous unitary one parameter groups $Q$ and $R$, we want to show that if $g_V(Q)=g_V(R)$ for all $V\in\mv(\mh)$ then $R=Q$. Now if $g_V(Q)=g_V(R)$ it follows that $\overline{\delta}(\hat{A}^Q)=\overline{\delta}(\hat{A}^R)=\overline{\delta}(\hat{A})$ are such that they satisfy conditions a) and b). 
Therefore, given the spectral family $\{\hat{E}_{\lambda}\}$ of $\hat{A}$, for each $\alpha(t)\in Q$, $(\hat{E}_{\lambda})_{\smile\smile}\{\alpha(t)\}$ and similarly, for each $\beta(t)\in R$, $(\hat{E}_{\lambda})_{\smile\smile}\{\beta(t)\}$. However this is precisely the condition for $\{\hat{E}_{\lambda}\}$ to be the spectral family of each $\alpha(t)$ and of each $\beta(t)$. It follows\footnote{Recall that the spectral family is uniquely specified by the operator it decomposes.} that $\alpha(t)=\beta(t)$.
\end{proof}
\begin{Corollary}
$f\circ g=id_{Sub_u(\ps{K})}$ and $g\circ f=id_{\ps{Ob}}$
\end{Corollary}

\begin{proof}
We want to show that $f$ and $g$ are inverse of each other. First of all we recall that the composition of injective maps is itself an injective map, thus both $f\circ g$ and $g\circ f$ are injective. We then consider the group $\ps{Q}_V$ and apply the composite map $f_V\circ g_V$ for any $V\in\mv(\mh)$, obtaining
\be
f_V\circ g_V(\ps{Q}_V)=f_V(\overline{\delta}(\hat{A}^Q)=\{e^{it\hat{A}^Q}|t\in\Rl\}
\ee
However, for the group $\{e^{it\hat{A}^Q}|t\in\Rl\}$ we obtain
\be
f_V\circ g_V(\{e^{it\hat{A}^Q}|t\in\Rl\})=f_V(\overline{\delta}(\hat{A}^Q)=\{e^{it\hat{A}^Q}|t\in\Rl\}
\ee
and since $f_V\circ g_V$ is injective it follows that $\{e^{it\hat{A}^Q}|t\in\Rl\}=\ps{Q}_V$.\\
On the other hand 
\be
g_V\circ f_V(\overline{\delta}(\hat{A})=g_V(\{e^{it\hat{A}}|t\in\Rl\})=\overline{\delta}(\hat{A})
\ee
\end{proof}

\section{Conclusion}

In this paper we have given a definition of the complex number quantity value object $\ps{\Cl^{\leftrightarrow}}$ in a topos.
The choice in the construction of $\ps{\Cl^{\leftrightarrow}}$ was motivated by the relations between the spectra of normal operators and the spectra of the self-adjoint operators comprising them. In particular, this newly defined object allowed us to define normal operators in the same way as self-adjoint operators were defined, namely as maps from the state space to the complex quantity value object.

In order to interpret these normal operators we defined them in terms of functions on filters, which we have called observable and antonymous functions. These then are related to the maximum and minumun value an individual normal operator can have.\\
We then analysed the way in which observable functions for normal operators are related to observable functions of the self-adjoint operators comprising them. 

Subsequently we have analysed the properties of the complex number value object and have found out that, similar to the real quantity value object, the complex quantity value object is only a monoid.
However, it is possible to turn both these objects into abelian groups via the process of k-extension. 
We utilised these abelian group objects to define the internal notion of one parameter groups in a topos. This enabled us to define the topos analogue of the Stone's theorem. This is very important when eventually analysing time evolution in the topos frame work. In fact, given the topos analogue of the Hamiltonian operator, via the Stone's theorem we can define a unique one parameter group of transformations which represent time evolution. The detailed analysis and consequences of this is left for future publication.

Moreover, when analysing $\ps{\Cl^{\leftrightarrow}}$ we showed that the results obtained in \cite{domain} for the real valued quantity value object can be easily generalised for the complex value object, this obtaining an interpretation of the $\ps{\Cl^{\leftrightarrow}}$ in terms of domain theory.

To apply this new topos framework to scenarios in quantum mechanics, the KMS state is a natural next step,
as the KMS condition requires complex quantities in order to be specified.
In particular, our definitions of one parameter group transformations should lead directly towards a specification of the KMS condition in topos quantum physics. This was done in \cite{kms}.

\bigskip
\bigskip

\textbf{Acknowledgements. }
One of the authors (C.F.) would like to thank her parents Luciano Flori and Elena Romani for their support and encouragement.
This work was supported by the Perimeter Institute of Theoretical Physics, Waterloo, Ontario and National Sciences and Engineering Research Council of Canada.

\newpage
\section{Appendix}
\subsection{Lattices}\label{sec:lattice}
In the previous section we have expressed the spectral theorem as referred to normal operators. We now will give a general definition of a spectral family as referred to a lattice $L$.
\begin{Definition}
Given a complete lattice $L$, a mapping 
\ba
E:\Rl&\rightarrow& L\\
\lambda&\mapsto&E_{\lambda}
\ea
is a spectral family in $L$ if the following hold
\begin{enumerate}
\item $E_{\lambda}\leq E_{\mu}$ for $\lambda\leq \mu$.
\item $E_{\lambda}=\bigwedge_{\mu>\lambda}E_{\mu}$ for all $\lambda\in \Rl$
\item $\bigwedge_{\lambda\in \Rl}E_{\lambda}=0$, $\bigvee_{\lambda\in\Rl}E_{\lambda}=1$
\end{enumerate}
If there exists $\alpha, \beta\in \Rl$ such that $E_{\alpha}=0$ for all $\lambda<\alpha$ and $E_{\lambda}=1$ for $\lambda\geq \beta$, then the spectral family $E$ is called bounded.
\end{Definition}
Given the definition of ordering of complex numbers given in (\ref{def:order}), we can trivially extend the above definition to a complex spectral measure as follows:
\begin{Definition}
Given a complete lattice $L$, a mapping 
\ba
E:\Cl&\rightarrow& L\\
\lambda&\mapsto&E_{\lambda}
\ea
which, from the definition of ordering (\ref{def:order}), is equivalent to 
\be
(\varepsilon, \eta)\mapsto E_{\varepsilon}E_{\eta}
\ee
This represents a spectral family in $L$ for ($\epsilon, \eta \in \Rl^2$) if the following holds
\begin{enumerate}
\item $E_{\varepsilon}E_{\eta}\leq E_{\varepsilon^{'}}E_{\eta^{'}}$ for $\varepsilon\leq \varepsilon^{'}$ and $\eta\leq\eta^{'}$.
\item $E_{\varepsilon}E_{\eta}=\bigwedge_{\varepsilon^{'}>\varepsilon}\bigwedge_{\eta^{'}>\eta}E_{\varepsilon^{'}}E_{\eta^{'}}$ for all $(\varepsilon,\eta)\in \Rl^2$
\item $\bigwedge_{\varepsilon\in\Rl}\bigwedge_{\eta\in \Rl} E_{\varepsilon}E_{\eta}=0$, $\bigvee_{\varepsilon\in\Rl}\bigvee_{\eta\in\Rl}E_{\varepsilon}E_{\eta}=1$
\end{enumerate}
If there exists $(\varepsilon^{'},\eta^{'})\in \Rl^2$ such that $E_{\varepsilon}E_{\eta}=0$ for all $(\varepsilon,\eta)<(\varepsilon^{'},\eta^{'})$ and $E_{\varepsilon}E_{\eta}=1$ for $(\varepsilon,\eta)\geq(\varepsilon^{'},\eta^{'})$, then the spectral family $E$ is called bounded.
\end{Definition}
If $L$ is a complete lattice and $a\in L$, then
\be
L_a:=\{b\in L|b\leq a\}
\ee
is a complete distributive lattice with maximal element $a$. Moreover, given a bounded spectral family $E$ in $L$, then 
\be
E^a:\lambda\mapsto E_{\lambda}\wedge a
\ee
is a spectral family in $L_a$.

In our case $L$ would be the lattice of projection operators in an abelian von Neumann algebra $V$ denoted by $P(V)$ and $E$ would be the spectral family in $P(V)$ of a normal operator $\hat{A}$. In this situation the restriction
\be
E^P:\lambda\mapsto E_{\lambda}\wedge P
\ee
is a bounded spectral family in the ideal
\be
I_P:=\{\hat{Q}\in P(V)|\hat{Q}\leq \hat{P}\}\subseteq P(V)
\ee
We now introduce the notion of a {\it filter} and a {\it filter base}.
\begin{Definition}
Given a lattice $L$ with zero element $0$, a subset $F$ of $L$ is called a (proper) filter (or (proper) dual ideal) if
\begin{itemize}
\item [i)] $0\notin F$.
\item [ii)] If $a, b\in F$, then $a\wedge b\in F$
\item [iii)] If $a\in F$ and $b\geq a$ then $b\in F$
\end{itemize}
\end{Definition}
If $F$ is such that there exists no other filter which contains it, then $F$ is a {\it maximal filter}.
Of particular importance is a {\it filter base}. 
\begin{Definition}
Given a lattice $L$ with zero element $0$, a subset $\mathcal{B}$ of $L$ is called filter base if
\begin{itemize}
\item [i)] $0\notin \mathcal{B}$.
\item [ii)] If $a, b\in \mathcal{B}$, then $\exists c\in F$ such that $c\leq a\wedge b$
\end{itemize}
\end{Definition}
If a filter base $ \mathcal{B}$ is such that there exists no other filter base which contains it, then $\mathcal{B}$ is a {\it maximal filter base}. For any lattice with a zero element, a maximal filter and maximal filter base always exist. Moreover one can deduce that, for a quasipoint $\mb$ of the lattice $L$, we have that 
\be 
\forall a\in\mb\;\forall b\in L \text{ if }a \leq b \text{ then }b\in \mb
\ee
Therefore 
\be
\forall a, b \in\mb;\; a\wedge b\in\mb
\ee
i.e. all maximal filter bases are maximal filters and vice versa.

A maximal filter base is what in the literature is called a {\it quasipoint}. For the sake of completeness we will report the definition below.
\begin{Definition}
A non empty subset $\mb$ of a lattice $L$ is a \textbf{quasipoint} in $L$ iff the following conditions are satisfied
\begin{enumerate}
\item$ 0\notin \mb$
\item $\forall a,b\in mb$ there exists a $c\in \mb$ such that $c\leq a\wedge b$
\item $\mb$ is a maximal subset having properties 1 and 2.
\end{enumerate}
\end{Definition}
It is easy to see that a maximal filter is nothing but a maximal dual ideal. In a complemented distributive lattice a maximal filter is called {\it ultra filter} and it has the property that either $a\in L$ or $a^c\in L$.

The set of all quasipoints in a lattice $L$ is denoted by $\mathcal{Q}(L)$.
Such a set can be given a topology whose basis sets are, for each $a\in L$
\be
\mathcal{Q}_a(L):=\{\mb \in\mathcal{Q}(L)|a\in\mb\}
\ee
We then have from the fact that $a\in \mb$ and $b\in\mb$ imply $a\wedge b\in\mb$ that
\be
\mathcal{Q}_a(L)\cap \mathcal{Q}_b(L)=\mathcal{Q}_{a\wedge b}(L)
\ee
From the fact that $0\notin \mb$ it follows that
\be
\mathcal{Q}_0(L)=\emptyset
\ee
Finally, since $I$ is the upper bound of the lattice $L$,
then
\be
\mathcal{Q}_I(L)=\mathcal{Q}(L)
\ee
Thus indeed the set $\mathcal{Q}_a(L)$ forms a basis. Moreover, from the property of maximality of quasipoints, it follows that the sets $\mathcal{Q}_a(L)$
are clopen. In particular, from the definition the sets $\mathcal{Q}_a(L)$ are open. To show that they are closed we need to show that they contain all their  limit points\footnote{Recall that given a set $S$ a point $x$ is called a limit point of $S$ iff for every open set containing $x$ it also contains another point of $S$ different from $x$.}. 
In this case a limit point of $S$ will be a quasipoint $\mb$ such that there exists a $b\in \mb$ for which $\mathcal{Q}_a(L)\cap\mathcal{Q}_b(L)\neq \emptyset$. So, to show that $Q_a(L)$ contains all of its limiting points we have to show that all points contained in the complement $\mathcal{Q}/\mathcal{Q}_a(L)$ will not satisfy the condition $\mathcal{Q}_a(L)\cap\mathcal{Q}_b(L)\neq \emptyset$. In particular, if $\mb\in \mathcal{Q}/\mathcal{Q}_a(L)$ then $a\notin \mb$ there exists a $b\in\mb$ such that $a\wedge b=0$, thus $\mathcal{Q}_a(L)\cap\mathcal{Q}_b(L)=\emptyset$.

The topology whose basis are the clopen sets $\mathcal{Q}_a(L)$ is Hausdorff zero dimensional. Given this topology we can now define what a Stone spectrum of a lattice is.
\begin{Definition}
$\mathcal{Q}(L)$ equipped with the topology whose basis sets are the clopen sets $\mathcal{Q}_a(L)$ is called the Stone spectrum of the lattice $L$.
\end{Definition}

It was shown in \cite{Groote2007a} that if the lattice $L$ is the lattice of projection operators in an abelian von Neumann algebra $V$, then the Stone spectrum coincides with the Gel'fand spectrum\footnote{Given an abelian von Neumann algebra $V$, the Gel'fand spectrum of $V$ consists of all the multiplicative linear functionals on $\lambda:V\rightarrow \Cl$ with values in the complex numbers, such that $\lambda(\hat{1})=1$. } of $V$.
\begin{Theorem}
Given an abelian von Neumann algebra $V$, the Gel'fand spectrum $\us_V$ of $V$ is homeomorphic to the Stone spectrum $\mathcal{Q}(P(V))$ of $V$.
\end{Theorem}
The proof can be found in \cite{Groote2007a} and rests on the fact that for each element $\lambda\in \us_V$ one can define the corresponding quasipoint
\be
\beta(\lambda):=\{\hat{P}\in P(V)|\lambda(\hat{P})=1\}
\ee
The mapping 
\ba\label{ali:stone}
\beta:\us_V&\rightarrow& \mathcal{Q}(P(V))\\
\lambda&\mapsto&\beta(\lambda)
\ea
are then the desired homeomorphisms.

\subsection{State Space and the Quantity Value Object}\label{sec:topos}
The topos analogue of the state space (\cite{Doering2008} ) is the object in $\Sets^{\mv(\mh)}$ called the spectral presheaf which is defined as follows:
\begin{Definition}
The spectral presheaf, $\Sig$, is the covariant functor from the
category $\mv(\mh)$ to $\Sets$ (equivalently, the
contravariant functor from $\mathcal{V(H)}$ to $\Sets$) defined
by:
\begin{itemize}
\item Objects: Given an object $V$ in $\mv(\mh)$, the associated set $\Sig(V)=\us_V$ is defined to be the Gel'fand spectrum of the (unital) commutative von Neumann sub-algebra $V$, i.e. the set of all multiplicative linear functionals $\lambda:V\rightarrow \Cl$, such that $\lambda(\hat{1})=1$.
\item Morphisms: Given a morphism $i_{V^{'}V}:V^{'}\rightarrow V$ ($V^{'}\subseteq V$) in $\mv(\mh)$, the associated function $\Sig(i_{V^{'}V}):\Sig(V)\rightarrow
\Sig(V^{'})$ is defined for all $\lambda\in\Sig(V)$ to be the
restriction of the functional $\lambda:V\rightarrow\Cl$ to the
sub-algebra $V^{'}\subseteq V$, i.e.
$\Sig(i_{V^{'}V})(\lambda):=\lambda_{|V^{'}}$.
\end{itemize}
\end{Definition}
On the other hand the quantity valued object is defined as follows:
In the topos $\Sets^{\mv(\mh)}$ the representation of the quantity value object $\mathcal{R}$ is given by the following presheaf:
\begin{Definition}
The presheaf $\ps{\Rl^{\leftrightarrow}}$ has as
\begin{enumerate}
\item [i)] Objects\footnote{A map $\mu:\downarrow V\rightarrow\Rl$ is said to be order reversing if $V^{'}\subseteq V$ implies that $\mu(V^{'})\leq\mu(V)$. A map $\nu:\downarrow V\rightarrow\Rl$ is order reversing if $V^{'}\subseteq V$ implies that $\nu(V^{'})\supseteq \nu(V)$.}: 
\be
\ps{\Rl^{\leftrightarrow}}_V:=\{(\mu,\nu)|\mu,\nu:\downarrow V\rightarrow\Rl|\mu\text{ is order preserving },\nu\text{ is order reversing }; \mu\leq\nu\}
\ee
\item[ii)] Arrows: given two contexts $V^{'}\subseteq V$ the corresponding morphism is
\ba
\ps{\Rl^{\leftrightarrow}}_{V,V^{'}}:\ps{\Rl^{\leftrightarrow}}_V&\rightarrow&\ps{\Rl^{\leftrightarrow}}_{V^{'}}\\
(\mu,\nu)&\mapsto&(\mu_{|V^{'}},\nu_{|V^{'}})
\ea
\end{enumerate}
\end{Definition}
This presheaf is where physical quantities take their values, thus it has the same role as the reals in classical physics.

The reason why the quantity value object is defined in terms of order reversing and order preserving functions is because, in general, in quantum theory one can only give approximate values to quantities.
In most cases, the best approximation to the value of a physical quantity one can give is the smallest interval of possible values of that quantity.
For details see \cite{Doering2008}.

\bibliography{jabref}{}
\bibliographystyle{hplain}

\end{document}